\documentclass[journal, onecolumn]{IEEEtran}
\usepackage{caption}
\usepackage{url}
\usepackage{ifthen}
\usepackage{cite}
\usepackage[cmex10]{amsmath} 
\usepackage{bbold}
\usepackage{bookmark}
\usepackage{tikz}
\usetikzlibrary{shapes.geometric, arrows, positioning}

\usepackage[capitalise]{cleveref}
\usepackage{balance}
\usepackage{amsmath}
\usepackage{amsthm}
\usepackage{amssymb}
\usepackage{multirow}
\usepackage{amsfonts}
\usepackage{graphicx}
\usepackage{algorithm}
\usepackage{url}
\usepackage{lipsum}
\usepackage{booktabs,makecell, multirow}
\usepackage{tikz}
\usetikzlibrary{matrix}
\usepackage{rotating}
\usepackage{mathtools}
\usepackage{tkz-tab}

\usepackage{latexsym}
\usepackage{color}
\usepackage{graphicx}
\usepackage{threeparttable}
\usepackage{float}

\usepackage{algpseudocode}

\usepackage{url}

\usepackage{chngpage}

\usepackage{tabularx}%

\usepackage{booktabs}
\usepackage{pbox}
\usepackage{orcidlink}

\usepackage{caption,subcaption}

\usepackage{mathtools}

\usepackage{titlesec}
\usepackage{hyperref}

\errorcontextlines\maxdimen

\newcounter{algsubstate}
\renewcommand{\thealgsubstate}{\alph{algsubstate}}

\DeclareGraphicsExtensions{.tif,.pdf,.jpeg,.png,.jpg,.bmp,.svg,.eps,.EMF}

\newtheorem{theorem}{Theorem}
\newtheorem{definition}{Definition}
\newtheorem{lemma}{Lemma}
\newtheorem{prop}{Proposition}
\newtheorem{corollary}{Corollary}
\newtheorem{remark}{Remark}

\newtheorem{claim}{Claim}

\newcommand{\cX}{\mathcal{X}}
\newcommand{\cY}{\mathcal{Y}}
\newcommand{\cM}{\mathcal{M}}

\newcommand{\cR}{\mathcal{R}}
\newcommand{\cS}{\mathcal{S}}
\newcommand{\cT}{\mathcal{T}}

\newcommand{\cH}{\mathcal{H}}

\newcommand{\bX}{\mathbf{x}}
\newcommand{\bY}{\mathbf{y}}

\newcommand{\bU}{\mathbf{u}}
\newcommand{\bV}{\mathbf{v}}
\newcommand{\bG}{\mathbf{g}}

\newcommand{\bZ}{\mathbf{z}}

\newcommand{\bPsi}{\mathbf{\psi}}

\newcommand{\bbE}{\mathbb{E}}

\newcommand{\poly}{\mathsf{poly}}
\newcommand{\dec}{\mathrm{dec}}

\newcommand{\supp}{\mathsf{supp}}

\newcommand{\simm}{\mathsf{sim}}

\newcommand{\mM}{\mathbf{M}}
\newcommand{\mT}{\mathbf{T}}
\newcommand{\mG}{\mathbf{G}}
\newcommand{\mA}{\mathbf{A}}
\newcommand{\mY}{\mathbf{Y}}
\newcommand{\mX}{\mathbf{X}}

\newcommand{\test}{\mathsf{test}}

\newcommand\norm[1]{\left\lVert#1\right\rVert}

\hypersetup{
  colorlinks=false, 
  pdfborder={0 0 1} 
}

\newcommand{\localcref}[1]{%
  \hyperref[#1]{%
    \llap{\hbox to 10mm{\hfill}}%
    \cref{#1}%
    \rlap{\hbox to -2mm{\hfill}}%
  }%
}

\begin{document}

\title{Matrix Completion in Group Testing: \\Bounds and Simulations}

\author{%
  \IEEEauthorblockN{Trung-Khang~Tran\orcidlink{0009-0007-4753-0948} and Thach V.~Bui\orcidlink{0000-0003-1368-8724}} \\
  \IEEEauthorblockA{Faculty of Information Technology, University of Science, Ho Chi Minh city, Vietnam \\
  Vietnam National University, Ho Chi Minh City, Vietnam \\
  Email: ttkhang2407@apcs.fitus.edu.vn, bvthach@fit.hcmus.edu.vn}
}


\maketitle

\thispagestyle{plain}
\pagestyle{plain}


\IEEEpeerreviewmaketitle

\begin{abstract}

The goal of group testing is to identify a small number of defective items within a large population. In the non-adaptive setting, tests are designed in advance and represented by a measurement matrix $\mM$, where rows correspond to tests and columns to items. A test is positive if it includes at least one defective item. Traditionally, $\mM$ remains fixed during both testing and recovery. In this work, we address the case where some entries of $\mM$ are missing, yielding a missing measurement matrix $\mG$. Our aim is to reconstruct $\mM$ from $\mG$ using available samples and their outcome vectors.

The above problem can be considered as a problem intersected between Boolean matrix factorization and matrix completion, called the matrix completion in group testing (MCGT) problem, as follows. Given positive integers $t,s,n$, let $\mY:=(y_{ij}) \in \{0, 1\}^{t \times s}$, $\mM:=(m_{ij}) \in \{0,1\}^{t \times n}$, $\mX:=(x_{ij}) \in \{0,1\}^{n \times s}$, and matrix $\mG \in \{0,1 \}^{t \times n}$ be a matrix generated from matrix $\mM$ by erasing some entries in $\mM$. Suppose $\mY:=\mM \odot \mX$, where an entry $y_{ij}:=\bigvee_{k=1}^n (m_{ik}\wedge x_{kj})$, and $\wedge$ and $\vee$ are AND and OR operators. Unlike the problem in group testing whose objective is to find $\mX$ when given $\mM$ and $\mY$, our objective is to recover $\mM$ given $\mY,\mX$, and $\mG$.

We first prove that the MCGT problem is NP-complete. Next, we show that certain rows with missing entries aid recovery while others do not. For Bernoulli measurement matrices, we establish that larger $s$ increases the higher the probability that $\mM$ can be recovered. We then instantiate our bounds for specific decoding algorithms and validate them through simulations, demonstrating superiority over standard matrix completion and Boolean matrix factorization methods.
\end{abstract}

\section{Introduction}
\label{sec:intro}

Group testing is a combinatorial optimization problem whose objective is to identify a small number of defective items in a large population of items efficiently~\cite{dorfman1943detection}. Defective items and non-defective (negative) items are defined by context. For example, in the Covid-19 scenario, defective (respectively, non-defective) items are people who are positive (respectively, negative) for Coronavirus. In standard group testing (GT), the outcome of the test on a subset of items is positive if the subset contains at least one defective item and negative otherwise. In this paper, we study the case in which some entries in the measurement matrix are missing and sets of input items and their corresponding test outcomes are observed (sampled). Our objective is to fully recover the measurement matrix by using these information.

Formally, we consider the matrix completion in group testing (MCGT) problem as follows. Let $\mA(i, :)$ and $\mA(:, j)$ be the $i$th row and $j$th column of matrix $\mA$, respectively. Given positive integers $t, s, n$, let $\mY = (y_{ij}) \in \{0, 1\}^{t \times s}$, $\mM := (m_{ij}) \in \{0, 1 \}^{t \times n}$, and $\mX := (x_{ij}) \in \{0, 1\}^{n \times s}$ such that
\begin{equation}
    \mY := \mM \odot \mX,
    \label{eqn:1st}
\end{equation}
where $y_{ij} := \mM(i, :) \odot \mX(:, j) := \bigvee_{k=1}^n (m_{ik} \wedge x_{kj})$ in which $\wedge$ and $\vee$ are AND and OR operators.

Suppose matrix $\mG := (g_{ij}) \in \{0, 1, \blacksquare \}^{t \times n}$ be a matrix generated from matrix $\mM$ by erasing some entries in $\mM$, where $\blacksquare$ is an erasure that cannot be determined to be 0 or 1. Specifically, when $\bar{\Psi} \subseteq [t] \times [n]$ is the set of missing entries in $\mG$, i.e., $g_{ij} = m_{ij}$ if $(i, j) \in [t] \times [n] \setminus \bar{\Psi}$ and $g_{ij} = \blacksquare$ if $(i, j) \in \bar{\Psi}$.

In order to facilitate understanding of the MCGT problem, we assume entries in $\mG$ and $\mM$ are generated from the Bernoulli distribution. In particular, for any $i \in [t]$ and $j \in [n]$, $\Pr(m_{ij} = 1) = p$, $\Pr(m_{ij} = 0) = 1 - p$, and $\Pr( (i, j) \in \overline{\Psi}) = q$, where $0 < p, q < 1$. Then each entry $g_{ij}$ in $\mG$ is generated as follows:
\begin{equation}
\Pr(g_{ij} = \blacksquare) = q; \Pr(g_{ij} = 1) = p(1-q); \Pr(g_{ij} = 0) =(1-p)(1-q).
\label{ran_prob}
\end{equation}

We also assume that the number of ones in each column in $\mX$ is exactly $d$. Our objective is to estimate the possibility of recovering $\mM$ given the missing matrix $\mG$, the input matrix $\mX$, and the outcome matrix $\mY$.

\subsection{Motivation}
\label{sec:intro:motive}

Although the BMF and matrix completion problems are applied in various applications, the MCGT problem might not well studied. Here we present two notable ones, which are cross-domain recommendation and connectome in neuroscience.

\subsubsection{Cross-domain recommendation}
\label{sec:intro:cross}

Let us consider a problem in cross-domain recommendation (transfer learning) in which user behavior or item features from one domain (e.g., movies) can be utilized to improve recommendations in another domain (e.g., books, music, or games). Its aim is to alleviate data sparsity or cold-start problems in a target domain by transferring knowledge from a source domain, where user preferences, item features are more abundant or reliable. This information may overlap partially or even be disjoint across domains. Consider the following problem which is illustrated in Fig.~\ref{fig:CrossDomain}. Suppose that there are $t$ users and $s$ items in the source domain. A user likes or dislikes an item. There is a set of $n$ features (or latent factors) that embodies for each user and item. A feature is shared (respectively, not shared) by the $i$th user and the $j$th item then the $i$th user likes (respectively, dislikes) the $j$th item. Mathematically, in the source domain, the rating user-item matrix $\mY^0 \in \{0, 1\}^{t \times s^0}$ can be decomposed by the Boolean product into two smaller matrices: a user-feature matrix $\mM^0 \in \{0, 1\}^{t \times n^0}$ and a feature-item matrix $\mX^0 \in \{0, 1\}^{n^0 \times s^0}$. In other words, we can write
\[
\mY^0 = \mM^0 \odot \mX^0,
\]
where $\mY^0(i, j) := \bigvee_{k=1}^{n^0} (\mM^0(i, k) \wedge \mX^0(k, j))$. In particular, $\mY^0(i, j) = 1$, i.e., user $i$ likes item $j$ if user $i$ and item $j$ share at least one feature, and $\mY^0(i, j) = 0$, i.e., user $i$ dislikes item $j$ otherwise.

\begin{figure}
    \centering
    \includegraphics[width=0.5\linewidth]{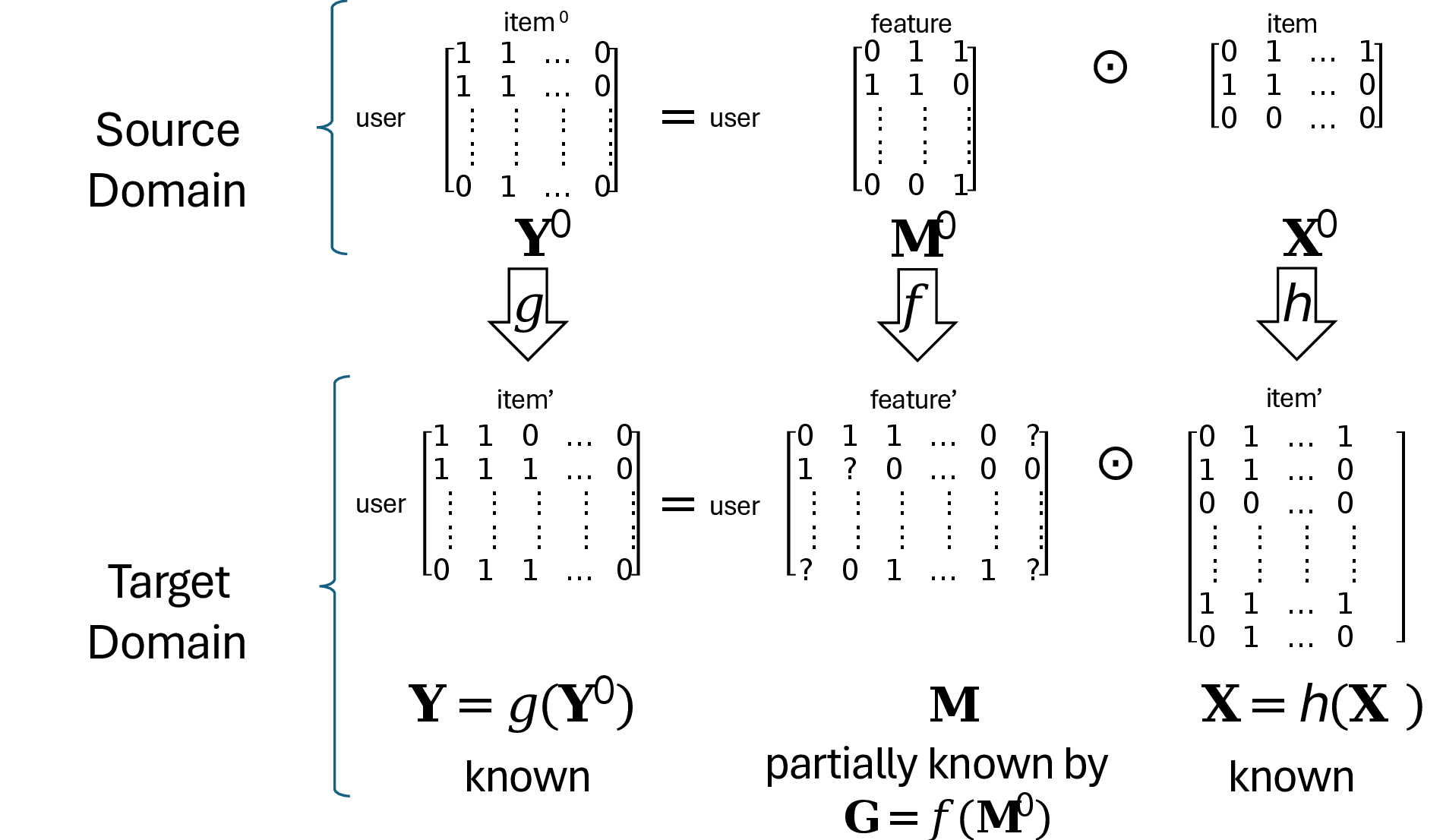}
    \caption{An example of cross-domain recommendation.}
    \label{fig:CrossDomain}
\end{figure}

Consider the target domain in which new items and features are introduced while the users remains unchanged. By using some transferring functions, one gets a new rating user-item matrix $\mY = g(\mY^0) \in \{0, 1\}^{t \times s}$ and a feature-item matrix $\mX = h(\mX^0) \in \{0, 1\}^{n \times s}$. However, instead of receiving a full user-feature matrix $\mM$ in the target domain, a partial user-feature matrix $\mG = f(\mM^0) \in \{0, 1, \blacksquare\}^{t \times n^0}$ is received in which two entries at row $i$ and column $j$ in $\mG$ and $\mM$ are identical, except for those being $\blacksquare$ in $\mG$. Our goal is to recover $\mM$ given $\mY, \mX$, and $\mG$.

\subsubsection{Connectome}

Building fully structural neuronal connectivity to better understand the structural-functional relationship of the brain is the main objective of connectome~\cite{sporns2011human}. For each person, building their fully structural neuronal connectivity when they are healthy can potentially help doctors treat them more easily when their brain does not function properly. Even in the case the doctors do not have their neuronal connectivity when they were healthy, having their neuronal connectivity when they are admitted to the hospital also helps them to identify causes by comparing it with other existing neuronal connectivities.

The most commonly used technique among them is functional Magnetic Resonance Imaging (fMRI), which offers an in-vivo view of both the brain's structure and function. Typically, there are three levels of resolutions: marco, meso, and micro. The macro level encompasses broad brain regions and the long-distance connections between them. The micro level focuses on cellular and neuronal details. To bridge the gap between the fine-grained details of individual neurons (micro level) and the more global connections between brain regions (macro level), the meso level is considered. At this level, one provides the network of connections between groups of neurons and local brain structures, such as cortical minicolumns and neural subnetworks. To construct a micro connectome, it is compulsory to know whether there exists a synapse between two neurons (the site where the axon of a neuron innervates to another neuron is called a synaptic site). A neuron that sends (respectively, receives) signals to another neuron across a synapse is called the presynaptic (respectively, postsynaptic) neuron. It is common to build connectome at meso or macro levels rather than a micro level because the brain is populated with roughly $100$ billion neurons~\cite{herculano2012remarkable} and this makes building a micro connectome nearly infeasible. However, with the recent development of neural population recordings, it is possible to record tens of thousands of neurons from (mouse) cortex during spontaneous, stimulus-evoked, and task-evoked epochs~\cite{stringer2019high,stringer2024rastermap}. This could enable the recording of most of the neurons in a brain region in the near future and thus could provide a sufficiently large number of observations with a set of presynaptic neurons spiked by an input stimulus and a set of postsynaptic neurons responded to the spiked presynaptic neurons, i.e., the input stimulus.

To construct a complete connectome of a brain, we present the neuron-neuron connectivities based on~\cite{bui2024simple} as follows. Let $\mT = (t_{ij})$ be an $(n + t) \times (n + t)$ connectivity matrix of $n + t$ neurons. Entry $t_{ij} = 1$ means there is a synaptic connection starting from neuron $i$ to neuron $j$, i.e., neuron $i$ is a presynaptic neuron and neuron $j$ is a postsynaptic neuron. On the other hand, $t_{ij} = 1$ means there is no synaptic connection starting from neuron $i$ to neuron $j$. Note that $\mT$ may not be symmetric. Since a stimulus can be stored as a memory at \emph{a few synapses}~\cite{smith2013dendritic,kastellakis2015synaptic,abraham2019plasticity}, a synapse can be determined by knowing which neurons participate in responding to the stimulus. Therefore, we can identify the synaptic connection between two neurons by dividing the neuron set into a set of \emph{a small number of presynaptic neurons} and a set of \emph{a large number of postsynaptic neurons} to observe their responses to a stimulus. In particular, we can create a $t \times n$ binary presynaptic-postsynaptic connectivity matrix $\mM = (m_{ij})$ as follows. Let $\cS \subseteq [n + t] = \{1, 2, \ldots, n + t \}$ with $|\cS| = t$ be a set of presynaptic neurons and $\overline{\cS} = [n + t] \setminus \cS$ with $|\overline{\cS}| = n$ be the set of postsynaptic neurons corresponding to the set of presynaptic neurons $\cS$. Matrix $\mM$ is obtained by removing every column $j \in S$ and every row $i \in \overline{\cS}$ in $\mM$. It directly follows that entry $m_{ij} = 1$ means there is a connection between the presynaptic neuron $i$ and the postsynaptic neuron $j$, and $m_{ij} = 0$ means otherwise.

Given a set of $n$ postsynaptic neurons labeled from $1$ to $n$, let $\cX \subseteq \{0, 1 \}^n$ be the discretized stimulus space (the ambient space) representing all stimuli. For any vector $\bV = (v_1, \ldots, v_n) \in \{0, 1 \}^n$, $v_j = 1$ means the postsynaptic neuron $j$ spikes and $v_j = 0$ means otherwise. Let $\bX = (x_1, \ldots, x_n) \in \cX$ be the binary representation vector for an input stimulus. Given a \emph{fixed} set of $t$ presynaptic neurons labeled $\{1, \ldots, t \} = [t]$, let $Y = (y_1, \ldots, y_t) \in \cY \subseteq \{0, 1 \}^t$ be the encoded vector for the input stimulus. For simplicity, we use a model proposed by Bui~\cite{bui2024simple} as follows: every presynaptic neuron spikes, and a postsynaptic neuron spikes if it does not connect to an inhibitory presynaptic neuron. Specifically, $y_i = 0$ means the presynaptic neuron $i$ is inhibitory and spiking, and $y_i = 1$ means the presynaptic neuron $i$ is excitatory and spiking. Note that every presynaptic neuron $i$ is a hybrid neuron, i.e., depending on the value of $y_i$, it can behave either as an excitatory neuron or an inhibitory neuron. A stimulus represented by $\bX$ is encoded at the $t$ presynaptic neurons as $\bY$. This model can be illustrated in Fig.~\ref{fig:Stimulus}.

\begin{figure}
    \centering
    \includegraphics[width=0.5\linewidth]{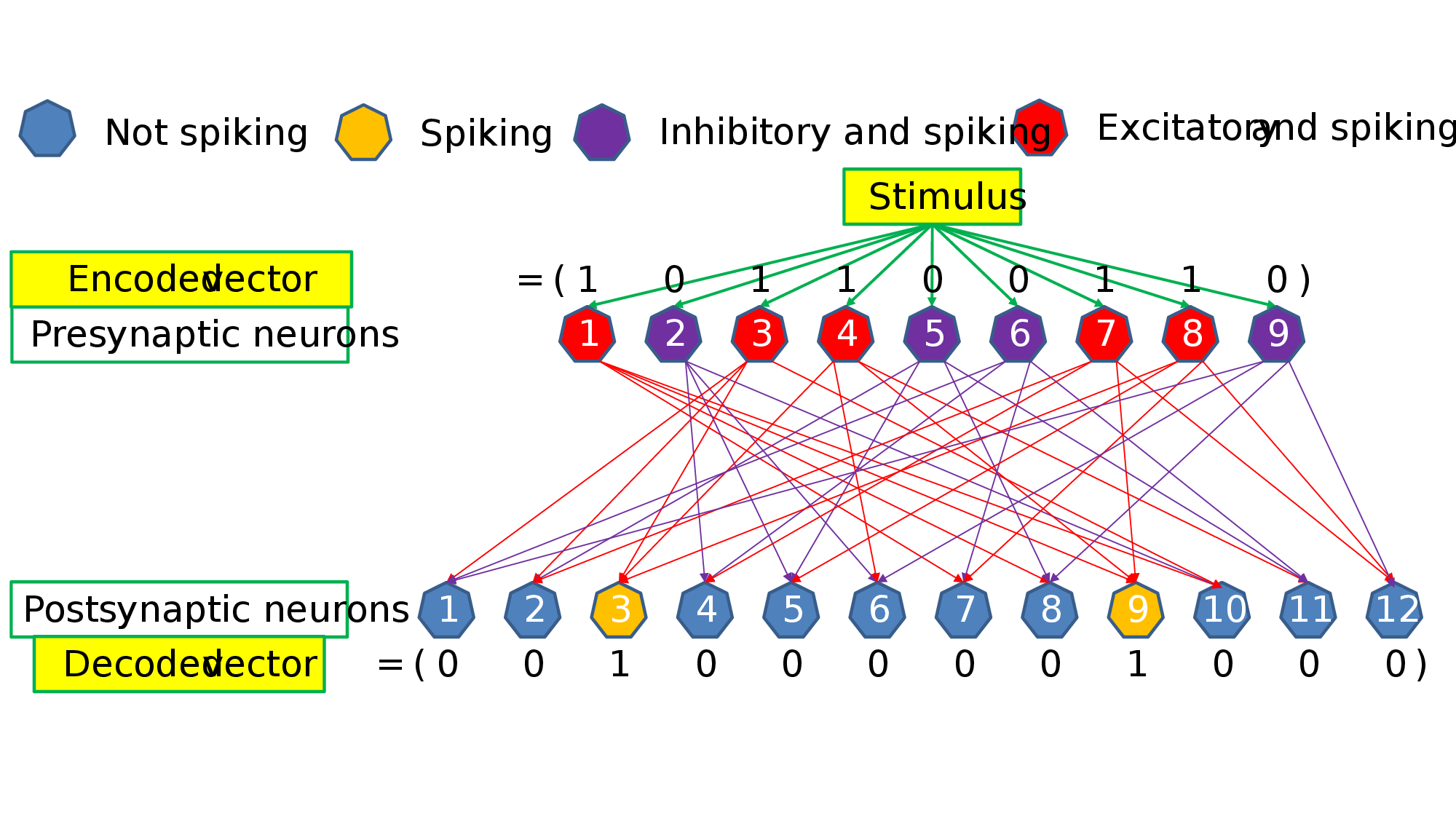}
    \caption{There are 21 neurons in total. They are divided into a set of 9 presynaptic neurons and a set of 12 postsynaptic neurons. The input stimulus is encoded by the presynaptic neurons and denoted as $\bY = (1, 0, 1, 1, 0, 0, 1, 1, 0)$ in which $y_i = 1$ means the $i$th presynaptic neuron is excitatory and spiking while $y_i = 0$ means the $i$th presynaptic neuron is inhibitory and spiking. The encoded stimulus is then decoded by the postsynaptic neurons and denoted as $\bX = (0, 0, 1, 0, 0, 0, 0, 0, 1, 0, 0, 0)$ in which $x_j = 1$ means the $j$th postsynaptic neuron spikes while $x_j = 0$ means the $j$th postsynaptic neuron does not spike.}
    \label{fig:Stimulus}
\end{figure}

The corresponding $t \times n$ binary presynaptic-postsynaptic connectivity matrix $\mM = (m_{ij})$ is:
\begin{equation}
    \mM = \left[ \begin{array}{cccccccccccc}
    0 & 0 & 0 & 0 & 0 & 1 & 1 & 1 & 1 & 0 & 0 & 0 \\
    0 & 0 & 0 & 1 & 1 & 1 & 0 & 0 & 0 & 1 & 0 & 0 \\
    1 & 1 & 1 & 0 & 0 & 0 & 0 & 0 & 0 & 1 & 0 & 0 \\
    0 & 0 & 1 & 0 & 0 & 1 & 0 & 0 & 1 & 0 & 1 & 0 \\
    0 & 1 & 0 & 0 & 1 & 0 & 0 & 1 & 0 & 0 & 1 & 0 \\
    1 & 0 & 0 & 1 & 0 & 0 & 1 & 0 & 0 & 0 & 1 & 0 \\
    0 & 1 & 0 & 0 & 0 & 0 & 1 & 0 & 1 & 0 & 0 & 1 \\
    0 & 1 & 0 & 0 & 1 & 0 & 1 & 0 & 0 & 0 & 0 & 1 \\
    1 & 0 & 0 & 0 & 0 & 1 & 0 & 1 & 0 & 0 & 0 & 1
    \end{array} \right], 
    \mG = \left[ \begin{array}{cccccccccccc}
    0 & 0 & 0 & 0 & 0 & 1 & 1 & 1 & 1 & 0 & 0 & 0 \\
    0 & 0 & \blacksquare & \blacksquare & 1 & 1 & 0 & 0 & 0 & 1 & 0 & 0 \\
    1 & 1 & 1 & 0 & 0 & 0 & 0 & 0 & 0 & 1 & 0 & 0 \\
    0 & 0 & 1 & 0 & 0 & 1 & 0 & 0 & 1 & 0 & 1 & 0 \\
    0 & 1 & \blacksquare & 0 & 1 & \blacksquare & 0 & 1 & 0 & 0 & 1 & 0 \\
    1 & 0 & 0 & \blacksquare & 0 & 0 & 1 & 0 & 0 & 0 & 1 & 0 \\
    0 & 1 & 0 & 0 & 0 & 0 & 1 & 0 & 1 & 0 & 0 & 1 \\
    0 & 1 & 0 & 0 & 1 & 0 & 1 & 0 & 0 & 0 & 0 & 1 \\
    1 & 0 & 0 & 0 & 0 & 1 & 0 & 1 & 0 & 0 & 0 & 1
    \end{array} \right] \label{eqn:connectivity}
\end{equation}

To distinguish two distinct stimuli, it is natural that any two distinct stimuli are represented by two distinct vectors in $\cX$ and are encoded by two distinct vectors in $\cY$. Then there exists a \emph{bijective} mapping from $\mM$ and $\bX$ to $\bY$ and let us denote it $\bY := \mM \odot \bX$, where $\odot$ is the mapping function. For $\bV = (v_1, v_2, \ldots, v_p)$, Let $\supp(\bV) := \{j \in [p] \mid v_j = 1 \}$ be the characteristic set of vector $\bV$. Then for any $j \in \supp(\bX)$, we must have $\supp(\mM(:, j) ) \subseteq \supp(\bY)$. Indeed, if there exists row $i_0$ such that $m_{i_0 j} = 1$ and $y_{i_0} = 0$, $x_j$ must equal to 0 because of the decoding rule. This contradicts the fact that $j \in \supp(\bX)$. Therefore, we get
\begin{equation}
    \bY = \bigvee_{j \in \supp(\bX)} \mM(:, j),
\end{equation}
and the mapping function $\odot$ is thus the testing operation in group testing.

Although the fast-paced development of neural recording could promise a complete connectome reconstruction, it is still impossible to identify some synaptic connections because of the obscured nature of these synapses. Let us denote these synapses \emph{missing synapses}. Let $r$ be the total number of unidentified synapses and $\overline{\Psi} = \{(i_1, j_1), \ldots, (i_r, j_r) \}$ be their corresponding set, where $1 \leq i_k \leq t$ and $1 \leq j_k \leq n$ for any $(i_k, j_k) \in \overline{\Psi}$. Let $\blacksquare$ be an erasure that cannot be determined to be 0 or 1. Then the presynaptic-postsynaptic connectivity matrix obtained by experiments is thus $\mG = (g_{ij})$ as illustrated in~\eqref{eqn:connectivity}, where $g_{ij} = m_{ij}$ if $(i, j) \not\in \overline{\Psi}$ and $g_{ij} = \blacksquare$ if $(i, j) \in \overline{\Psi}$. Note that for any stimulus $\bY$, one always receives $\bX := \dec(\mG, \bY)$ though we do not know every entry in $\mG$. More importantly, because of spontaneous neural activity, we cannot control stimulus inputs as we wish because we do not know which spiking neurons represent a stimulus in general. In other words, it is infeasible to generate a stimulus that induces the corresponding $\bX$ or $\bY$ as wanted. The problem of reconstructing the complete connectome turns out to be the problem of reconstructing $\mM$ from $\mG$ by collecting enough pairs $(\bX, \bY)$, i.e., $\mX$ and $\mY$ in~\eqref{eqn:1st} are large enough, in group testing.

\subsection{Related work}
\label{sec:intro:related}
MCGT can be considered as a problem that lies between Boolean matrix factorization and matrix completion. In standard Boolean matrix factorization~\cite{miettinen2021recent}, given $\mY$ and integer $n$, ones tries to find matrices $\mM$ and $\mX$ such that they minimize 
\begin{equation}
    \norm{\mY - \mM \odot \mX}^2_F = \sum_{i, j} y_{ij} \oplus (\mM \odot \mX)_{ij},
    \label{eqn:BMF}
\end{equation}
where subtraction, Frobenius norm $\norm{\cdot}_F$, and the summation are taken over the standard algebra while $\oplus$ stands for the element-wise XOR-operation. In standard matrix completion, a partial part of $\mM$ (or $\mX$ or $\mY$), says $\mG$ is given, the objective is to recover $\mM$ such that $\mbox{rank}(\mM)$ is minimized subject to $m_{ij} = g_{ij}$ for $(i, j) \in [t] \times [n] \setminus \bar{\Psi}$.

\subsubsection{Group testing} In standard group testing, one attempts to recover $\mX$ given $\mM$ and $\mY$ for $s = 1$ in~\eqref{eqn:1st}. 

Let $\supp(\bV) = \{j \mid v_j \neq 0 \}$ be the support set for vector $\bV = (v_1, \ldots, v_w)$ and $\cM_i = \supp(\mM(i, :) )$ for $i = 1, \ldots, t$. The OR-wise operator between two vectors of same size $\bZ = (z_1, \ldots, z_n)$ and $\bZ^\prime = (z_1^\prime, \ldots, z_n^\prime)$ is $\bZ \vee \bZ^\prime = (z_1 \vee z_1^\prime, \ldots, z_n \vee z_n^\prime)$. Let $\test(\cdot)$ be the notation for the test operations in group testing; namely, $y_i := \mM(i, :) \odot \bX := \test(\cM_i \cap \supp(\bX)) = 1$ if $|\cM_i \cap \supp(\bX)| \geq 1$ and $y_i = 0$ if $|\cM_i \cap \supp(\bX)| = 0$, for $i = 1, \ldots, t$. The procedure to get outcome vector $\bY$ is called \textit{encoding} and the procedure to recover $\bX$ from $\bY$ and $\mM$ is called \textit{decoding}.

There are several criteria for tackling group testing, but we focus on four main ones here. The first criterion is the testing design, which can be either non-adaptive or adaptive. In a non-adaptive design, all tests are predetermined and independent, allowing them to be executed in parallel to save time. In contrast, an adaptive design involves tests that depend on the previous tests, often requiring multiple stages. While this design can achieve the information-theoretic bound on the number of tests, it tends to be time-consuming due to the need for multiple stages. The second criterion is the setting of the defective set. In a combinatorial setting, the defective set is arbitrarily organized subject to predefined constraints, whereas in a probabilistic setting, a distribution is applied to the input items. The third important criterion is whether the design is deterministic or randomized. A deterministic design produces the same result given the same inputs, whereas a randomized design introduces a degree of randomness, which may lead to different results when executed multiple times. Finally, the fourth criterion involves the recovery approach: exact recovery, where all defective items are identified, and approximate recovery, where only some of the defective items are identified.

Since the inception of group testing, it has been applied in various fields such as computational and molecular biology~\cite{du2000combinatorial}, networking~\cite{d2019separable}, and Covid-19 testing~\cite{shental2020efficient}. For combinatorial group testing with non-adaptive designs, a strong factor of $d^2$ has been established in the number of tests \cite{kautz1964nonrandom,d1982bounds,dyachkov1983survey}. For exact recovery, it is possible to obtain $O(d^2 \log n)$ tests that can be decoded in time $O(tn)$ with explicit construction~\cite{porat2011explicit} or in $\poly(d, \log{n})$ \cite{indyk2010efficiently,ngo2011efficiently,cheraghchi2013noise,cheraghchi2020combinatorial} with additional constraints on construction. To reduce the factor $d^2$ to $d$, an adaptive design or a probabilistic setting can be used. The set of defective items can be fully recovered by using $O(d \log{(n/d)})$ tests with $O(\log_d{n})$ stages in~\cite{du2000combinatorial} or with two stages in~\cite{de2005optimal}. When the test outcomes are unreliable, it is still possible to obtain $O(d \log{(n/d)})$ tests using a few stages~\cite{scarlett2018noisy,teo2022noisy,hwang1975generalized,mezard2011group,aldridge2020conservative}. For probabilistic group testing with non-adaptive design, the number of tests $O(d \log{n})$ has been known for a long time\cite{aldridge2014group,scarlett2016phase,aldridge2019group,coja2020optimal}. Decoding time associated with that number of tests has gradually reduced from $\poly(d, \log{n})$ to near-optimal $O(d \log{n})$~\cite{cheraghchi2020combinatorial,price2020fast}. Many variants of group testing such as threshold group testing~\cite{damaschke2006threshold}, quantitative group testing~\cite{bshouty2009optimal}, complex group testing~\cite{chen2009nonadaptive}, concomitant group testing~\cite{bui2023concomitant}, and community-aware group testing~\cite{nikolopoulos2023community} has also been considered recently. However, to the best of our knowledge, all of these models are not closely related to our setup.

\subsubsection{Matrix completion} A closely related research topic to our work is matrix completion which was first known as \emph{Netflix problem}~\cite{acm2007proceedings}. In this problem, Netflix database consists of about $t \approx 10^6$ users and about $n \approx 25,000$ movies with users rating movies. Suppose that $\mM_{t \times n}$ is the (unknown) users rating matrix that we are seeking for. Since most of the users have only seen a small fraction of the movies, only a small subset of entries in $\mM$ have been identified and the rest are considered as missing entries. The actual ratings are recorded into matrix $\mG = (g_{ij}) \in \{\mathbb{R} \cup \{ \blacksquare \} \}^{t \times n}$. The goal is to predict which movies a particular user might like. Mathematically, we would like to complete matrix $\mG$, i.e., replacing missing entries by users rates, based on the partial observations of some of its entries to reconstruct $\mM$. Once $\mM$ is low-rank, it is possible to complete the matrix and recover the entries that have not been seen with high probability~\cite{candes2012exact}. In particular, if $\mbox{rank}(\mM) = k$, $n^* = \max(t, n)$, and each entry is observed uniformly, then there are numerical constants $C$ and $c$ such that if the number of observed entries is at least $C (n^*)^{5/4} k \log{n^*}$, all missing entries in $\mG$ can be recovered with probability at least $1 - cn^{-3} \log{n}$. Following this pioneering work, there is much work to tackle this problem with the same settings or different settings~\cite{keshavan2010matrix,candes2010power,recht2011simpler,bhojanapalli2014universal}. Unfortunately, the results in~\cite{candes2012exact} are inefficient when $k = O(t)$ because every entry must be observed in that case. Moreover, it is not utilized whether the missing entries are zero or non-zero. Although recovering measurement matrices in group testing is equivalent to the matrix completion problem, the settings in group testing are different from the settings in the standard matrix completion problem. Specifically, operations in group testing are Boolean and the test outcomes provide additional information compared to the matrix completion problem itself.

\subsubsection{Boolean Matrix Factorization} Given solely $\mY$, one tries to find matrices $\mM$ and $\mX$ satisfying the conditions in~\eqref{eqn:BMF}. Unfortunately, those matrices might not be found efficiently because the BMF problem is NP-complete~\cite{orlin1977contentment}. Therefore, using some solution in BMF to find $\mM$ without using the aid of $\mG$ and $\mX$ might not yield the actual $\mM$.

\subsection{Notation}
\label{sec:intro:formula}

For consistency, we use capital bold letters for matrices, non-capital letters for scalars, bold letters for vectors, and calligraphic letters for sets. We have summarized the commonly used notations in this paper and their definitions in \cref{tab: Notation}.

Let $\Psi=\{\Psi_1, \Psi_2, \dots, \Psi_r\} \subseteq [tn]$ be the set of missing entries from left to right and from top to bottom in $\mM$. In particular, for any $\Psi_g \in \Psi$, if $\Psi_g = (i_g-1) n +j_g$ where $i_g \in \{1, \dots, t\}$ and $j_g \in \{1,\dots, n\}$ then $\Psi_g$ is located at row $i_g$ and column $j_g$. Let $\overline{\Psi} = \{(i_1, j_1), \ldots, (i_r, j_r) \}$ be the bijective mapping set of $\Psi$, where $a_g = (i_g-1) n +j_g$ is represented by the pair $(i_g, j_g)$ and vice versa for $g \in [r]$. Let $\bPsi = (\psi_1, \ldots, \psi_r)^T \in \{0, 1 \}^r$ be the characteristic vector of $\overline{\Psi}$ in which $\psi_g = m_{i_g j_g}$ for $(i_g, j_g) \in \overline{\Psi}$.

Let $\cT_d := \{ \bX \in \{0, 1 \}^n \mid |\bX| = d \}$ be the set of all binary vectors of length $n$ and weight $d$. Set $\mX := [\bX_1, \bX_2, \ldots, \bX_s]$ and $\mY := \mM \odot \mX$, where $\bY_j = \mY(:, j) := \mM \odot \mX(:, j) = \mM \odot \bX_j$ for $j = 1, \ldots, s$. We further denote $\operatorname{col}(\mX) := \left\{\mX(:,j) |j\in\{1,\dots,s\}\right\}$, the set of all column vectors of $\mX$.

\begin{table}[t]
\centering

\scalebox{0.9}{
\begin{tabular}{|c|c|}
\hline
\textbf{Notation} & \textbf{Definition}  \\
\hline
[n]& The set $\{1,\dots,n\}$ \\
\hline
$\mathcal{D}$ & Defective set \\
\hline
d & Cardinality of the Defective set \\
\hline
$t$ & Number of test \\
\hline
$n$ & Number of item \\
\hline
$s$ & Number of sample \\
\hline
$p$ & Probability of a cell in $\mM$ having the value 1 \\
\hline
$q$ & Probability of a cell in $\mM$ goes missing \\
\hline
$h$ & Number of rows of the missing matrix \\
\hline
$\sigma_i$ & Number of missing cells in row $i$ of $\mM$ \\
\hline
$\varphi_i$ & Number of cells having value 1 in row $i$ of $\mM$ \\
\hline
$\mathcal{R}_i$ & The matrix with row from $(s(i-1)+1)^{th}$ to $si^{th}$ and columns from $(\sum_{0< j < i}\sigma_j+1)^{th}$ to $(\sum_{0< j < i+1}\sigma_j)^{th}$ of $\Gamma$ \\
\hline
$\bV_i$ & The slice of the vector $\bV$ from position $(s(i-1)+1)^{th}$ to $si^{th}$ \\
\hline
$\Psi[i]$ & The slice of the vector $\Psi$ from position $(\sigma_1+\dots+\sigma_{i-1}+1)^{th}$ to $(\sigma_1+\dots+\sigma_i)^{th}$ \\
\hline
$\bar{\varphi}_i$ & the number of $\Psi_\alpha\in\Psi[i]$ such that $\psi_\alpha=1$ \\
\hline
$\mathcal{Q}_i$ & The group testing problem containing the measurement matrix $\mathcal{R}_i$, the output vector $\bV_i$ and the sample vector $\Psi[i]$ \\
\hline
$\mathbf{M}$ & Measurement matrix of shape $t\times n$ \\
\hline
$\mathbf{G}$ & Measurement matrix of shape $t\times n$ with missing cells \\
\hline
$\mX := [\bX_1, \bX_2, \ldots, \bX_s]$ & The sample matrix \\
\hline
$\mY := \mM \odot \mX$& The outcome matrix \\
\hline
$\operatorname{col}(\mathbf{X}):=\left\{\mathbf{X}(:,j)^{\top}|j\in\{1,\dots,s\}\right\}$& The set of all column vectors of $\mathbf{X}$, sometimes is referred as the sample set \\
\hline
$\Gamma$& The converted missing matrix \\
\hline
$\bV$& The converted missing vector \\
\hline
$\blacksquare$ & An erasure that cannot be determined to be 0 or 1 \\
\hline
$\Psi=\{\Psi_1, \Psi_2, \dots, \Psi_r\}$ &  The set of missing entries from left to right and from top to bottom in $\mM$\\
\hline
$\overline{\Psi} = \{(i_1, j_1), \ldots, (i_r, j_r) \}$ & The bijective mapping set of $\Psi$\\
\hline
$\bPsi = (\psi_1, \ldots, \psi_r)^T $& The characteristic vector of $\overline{\Psi}$\\
\hline
$\cT_d := \{ \bX \in \{0, 1 \}^n \mid |\bX| = d \}$& The set of all binary vectors of length n and weight d\\
\hline
$a$& $1-p+pq$\\
\hline
$b$& $(1-p)(1-q)$\\
\hline
$\simm(\mathbf{x}_i,\mathbf{x}_j)$& The number of row that has the value 1 in both $\bX_i$ and $\bX_j$\\
\hline
$\Upsilon(u)$& $\left[(1-p)(1-q)\right]^{2d-2u}(1-p+pq)^{u}-\left[(1-p)(1-q)\right]^{2d-u}$\\
\hline
$\Pr(succ)$& The success probability of solving the Group testing problem via the recovered matrix\\
\hline
\end{tabular}}
\caption{Commonly used notations in this paper}
\label{tab: Notation}
\end{table}

\subsection{Contributions}
\label{sub:intro:contri}

We consider a model that lies between Boolean matrix factorization and matrix completion called matrix completion in group testing. In particular, given a missing matrix $\mG$, the input matrix $\mX$, and the outcome matrix $\mY$, we construct a converted missing matrix $\Gamma$ and a converted missing vector $\bV$  such that $\Gamma \odot \bPsi = \bV$ with no duplicated rows in $\Gamma$ (more details can be found in \cref{thm:construction}). Here, $\Gamma$ combined with $\mathbf{v}$ can be interpreted as an intermediate group testing problem, where the item set corresponds to the missing entries. Overall, our contribution is in five folds.
\begin{enumerate}
    \item We show that the information gain from the converted missing matrix $\Gamma$ and converted missing vector $\bV$ is equivalent to that of the missing matrix $\mG$ and the set of samples $\operatorname{col}(\mathbf{X})$. Therefore, to reconstruct the matrix $\mM$, one only needs to reconstruct $\bPsi$ from $\Gamma$ and $\Gamma \odot \bPsi$. This would reduce the complicated matrix completion problem into a standard group testing problem. Moreover, due to the nature of the standard group testing problem, the MCGT problem can be proved to be NP-complete.
    \item Since each entry in $\Gamma$ is independent and identically distributed, the more rows $\Gamma$ has, the better the chance we have of recovering $\bPsi$. Thus, to demonstrate the effectiveness of our approach, we analyze the number of rows in $\Gamma$. We do this by first deriving a bound for the expected number of rows, denoted as $h$, in $\Gamma$ as follows:\begin{equation}
    s\Upsilon(d)\left[1-\dfrac{s-1}{2d} \cdot \dfrac{a^2}{b-a^2}\right] \leq \frac{\bbE[h]}{t}  \leq s\Upsilon(d).
\end{equation}
where $a=(1-p)(1-q), b=1-p+pq$, $\Upsilon(d)=b^d-a^d$, and $t$ is the number of tests of the measurement matrix $\mM$. We then prove a concentration bound between \( h \) and its expected value (more details can be found in \cref{theo concentration bound}), which shows that as \( d \) and \( n \) approach infinity, we have
\[
\Pr(|h - \mathbb{E}[h]| > \delta) \to 0, \quad \forall \delta > 0.
\]
\item One aspect we wish to analyze is the performance of the recovery matrix as a measurement matrix in the initial group testing problem, compared to the original matrix. By redefining the construction of $\Gamma$ and $\mathbf{v}$, we derive a method to split the group testing between $\Gamma$ and $\mathbf{v}$ into smaller group testing problems with a Bernoulli-generated measurement matrix. Due to the well-known properties of Bernoulli-generated group testing problems, we are then able to establish a lower bound for the success rate of solving the group testing problem via the recovery matrix with any given algorithm $\mathcal{A}$. We then apply this general bound to derive specific bounds for conventional group testing algorithms, including COMP, DD, SCOMP, and SSS.
\item After deriving the bound for the success probability of solving the group testing problem via the recovery matrix, we investigate how this bound changes as we increase the number of sampling rounds (i.e., increase the number of matrix $\mX$). To address this, we derive \cref{theo: universal bound for GT}, which provides a universal lower bound for the success probability of any given algorithm $\mathcal{A}$ for recovery and $\mathcal{B}$ for solving the group testing problem using the recovery matrix. Furthermore, if both $\mathcal{A}$ and $\mathcal{B}$ guarantee a positive success probability for all group testing setups, one can show that as the number of sampling rounds increases, the lower bound for the success probability converges to 1. Building on this, we also derive \cref{last cor}, which specifies the number of sampling rounds required to guarantee a certain success probability.
\item To demonstrate the effectiveness of our proposed method, we conduct simulations on a standard noiseless group testing setup with a Bernoulli test design, where 10\% of the cells in the measurement matrix are allowed to go missing. We recover the measurement matrix using our method, Nuclear Norm Minimization (NNM), a recovery method in matrix completion, and GreConD, a recovery method in Boolean Matrix Factorization. We compare the accuracy of the group testing problem using the measurement matrix recovered by each of the three methods, as well as the initial true matrix. The results show that our proposed recovery matrix outperforms those recovered by NNM and GreConD, and it closely approximates the recovery achieved by the original matrix.

\end{enumerate}

\section{Information gain between missing and outcome matrices}
\label{sec:infoGain}

\par In this section, we show that solving the MCGT problem is equivalent to recovering the missing entries.

\subsection{Construction}
\label{sec:infoGain:construction}

\par To calculate the information gain between $\mX,\mY$ and the missing matrix $\mG$, we define an \emph{informative pair} of an column in $\mX$ and a row in $\mG$ as follows:

\begin{definition}[Informative pair] 
  Let $\mX$, $\mG$, and $t$ be defined in~\cref{sec:intro}. For any vector $\bX = (x_1, x_2, \ldots, x_n) \in \operatorname{col}(\mathbf{X})$ and $i \in [t]$, we say that the pair $(\bX, i)$ is informative if both of the following conditions are satisfied:
    \begin{itemize} 
    \item $\exists j_0 \in \supp(\bX)$, $g_{i j_0} = \blacksquare$,
    \item $\forall j \in \supp(\bX)$, $g_{ij} = 0$ or $g_{ij} = \blacksquare$.
    \end{itemize}
\label{informative-pair}
\end{definition}

If item $j$ is non-defective, the true value of any missing entry in the representative column of item $j$, i.e., $\mG(:, j)$,  does not affect the test outcomes. On the other hand, if item $j_0$ is defective and satisfies the first condition, the outcome of test $i$ may be affected by item $j_0$. The second condition ensures that row $i$ never contains a defective item $j^\prime$ and the entry of that item, i.e., $g_{i j^\prime}$, is known to be 1. Otherwise, test $i$ is positive, and any missing entry in row $i$ except $g_{i j^\prime}$ does not affect the outcome of test $i$. For example, let us consider a column of $\mX$ as follow 
\begin{align*}
    \bX=[0,1,1,0,0].
\end{align*}

and three test outcomes:
\begin{align*}
    \mathbf{y}_1= [0,0,0,1,\blacksquare], \mathbf{y}_2= [0,\blacksquare,1,0,0],\mathbf{y}_3= [1,0,1,0,1].
\end{align*}

\par Here it can be seen that $(\bX, 1)$ is not informative because it violates the first condition of Definition~\ref{informative-pair}. The pair $(\bX, 2)$ is also not an informative pair because it violates the second condition of Definition~\ref{informative-pair}. The pair $(\bX, 3)$, however, satisfies both conditions of Definition~\ref{informative-pair}. Hence, it is informative.

Based on the definition of informative pairs, we proceed to define a converted missing matrix and a converted missing vector in order to recover the missing matrix $\mG$. The procedure to construct $\Gamma$ and $\bV$ is described in words in \cref{thm:construction} and is written more formally in pseudocode in Algorithm~\ref{alg:missing_matrix_vector} at Appendix~\ref{app:alg}.

\begin{algorithm}
\caption{$\mathrm{missingBMF2GT}(\mG,\mX,\mY,t)$}
\label{thm:construction}

\begin{algorithmic}[1]
\State \textbf{Input:} Missing matrix $\mG$, sample matrix $\mathbf{X}$, outcome matrix $\mY$, number of tests $t$.
\State \textbf{Output:} Converted missing matrix $\Gamma$, converted missing vector $\mathbf{v}$.
\State Let $t, n, r, \mathbf{X}, \bPsi$ be defined in~\cref{sec:intro}. We begin with a $0\times r$ matrix $\Gamma$ and a $0\times 1$ vector $\bV$.
\State For every pair $(\bX, c) \in \operatorname{col}(\mathbf{X}) \times [t]$ that is informative we add a row vector $\mathbf{g} = (g_1, \ldots, g_r)$ to $\Gamma$.
\State For every $z \in [r]$ and $(i_z, j_z) \in \overline{\Psi}$, we set $g_z = 1$ if $i_z= c$ and $x_{j_z}= 1$, otherwise, we set $g_z=0$.
\State We append to $\bV$ one more entry, set to $y_c$. If there are two rows in $\Gamma$ that are the same, we delete one of them.
\State Repeat over all pairs of $(\bX,c)$.
\end{algorithmic}
\end{algorithm}

\begin{corollary}
    For $\mG,\mX,\mY\text{ and }t$ defined \cref{tab: Notation}, let $\Gamma$ and $\bV$ be the output of $\mathrm{missingBMF2GT}$($\mG,\mX,\mY,t)$, then we have $\Gamma \odot \bPsi = \bV$, and the time to construct $\Gamma$ and $\bV$ is $O(rs^2t^2 + stn)$.
\label{cor:GammaV}
\end{corollary} 

\begin{proof}
The equation $\Gamma \odot \bPsi = \bV$ is straightforwardly obtained. Since $\mG$ has a size of $t \times n$, there are up to $t$ rows that have missing entries. On the other hand, it takes $O(r (st)^2)$ time to check duplicated rows in $\Gamma$ and there are $s$ pairs $(\bX, \bY := \mM \odot \bX)$ observed, it takes $O(tns) + O(r (st)^2) = O(rs^2t^2 + stn)$ time to construct $\Gamma$ and $\bV$.
\end{proof}

For example, consider the missing matrix $\mG$ as in~\eqref{eqn:connectivity}. Suppose that $\mathbf{X}$ and $\mathbf{Y}$ are given as follow:
\begin{align}
\mathbf{X} &:= \left[ \bX_1 := [0,0,1,1,0,0,0,0,0,0,0,0]^T, \bX_2:=[0, 0,1,0,0,1,0,0,0,0,0,0]^T\right]; \\
\mathbf{Y}&:=  \left[\bY_1 := \mM \odot \bX_1 = [0,1,0,0,1,1,0,1,0]^T, \bY_2:= \mM \odot \bX_2 = [1,1,1,1,0,0,0,0,1]^T\right].
\end{align}

As described in~\cref{sec:intro}, the set of missing entries is $\Psi:=\{\Psi_{1}=15, \Psi_{2}=16, \Psi_{3}=51, \Psi_{4}=54, \Psi_{5}=64\}$ and the set of positions of those entries is $\overline{\Psi} = \{(2, 3), (2,4),(5,3),(5,6), (6, 4) \}$. By applying~\cref{informative-pair}, the following pairs are informative: $(\bX_1,2)$, $(\bX_1,5)$, $(\bX_1,6)$, $(\bX_2,5)$. Thanks to~\cref{thm:construction}, the converted missing matrix $\Gamma$ and converted missing vector $\bV$ can be constructed as follows:
\begin{equation}
    \Gamma = \left[ \begin{array}{ccccc}
    1 & 1 & 0 & 0 & 0 \\
    0 & 0 & 1 & 0 & 0 \\
    0 & 0 & 0 & 0 & 1 \\
    0 & 0 & 1 & 1 & 0 \\
    \end{array} \right], 
    \bV = \left[ \begin{array}{c}
    1\\
    0\\
    1\\
    0    
    \end{array} \right].
\end{equation}
Our main goal is to recover the vector $\psi = (\psi_1, \psi_2, \psi_3, \psi_4, \psi_5)$, which is $(0, 1, 0, 0, 1)$.

\subsection{Information equivalence between measurement and missing matrices}
\label{validation}

\par In this section, we will show that the information gain from the converted missing matrix $\Gamma$ and converted missing vector $\bV$ is equivalent to that from the missing matrix $\mG$ and matrix $\mathbf{X}$. To demonstrate this, we need to show that for a given $\bX \in \operatorname{col}(\mathbf{X})$, if $(\bX, i)$ is not informative, then there is no information gain on $\Psi$. This is equivalent to stating that if there is information gain between $\Psi$ and the test outcome on test $i$ with respect to the column vector $\bX$, the pair $(\bX, i)$ must be informative. We summarize this argument in the following theorem.

\begin{theorem}
    Let $t, \mathbf{X}, \Psi$ be variables defined in~\cref{sec:intro}. Given $\bX \in \operatorname{col}(\mathbf{X})$ and $\bY := (y_1, \ldots, y_t)^T := \mM \odot \bX$, for any $i \in [t]$ such that $(\bX, i)$ that is not informative, we have:
    \begin{equation}
        I(\Psi; y_i) = 0. \nonumber
    \end{equation}
\label{thm:mutual_info}
\end{theorem}
\begin{proof}

To prove this theorem, we first define $\overline{\cX}_i := { (i, j) \mid j \in \supp(\mG(i, :)) \cap \supp(\bX) }$ to be the set of corresponding entries of the defective items in row $i$. Then, if a pair $(\bX, i)$ is not informative, we have
\begin{equation}
    \Pr(y_i = 0) = \prod_{(i, j) \in(\overline{\cX}_i \setminus \overline{\Psi})} \Pr(g_{ij}=0).
    \label{eqn:not_info}
\end{equation}

Indeed, based on~\cref{informative-pair}, if $(\bX, i)$ is not informative, one of the two following possibilities must happen:
\begin{itemize}
\item For all $j_0 \in \supp(\bX)$, $(i, j_0) \not\in \overline{\Psi}$. It is equivalent that row $i$ does not have any missing entry. In other words, $\overline{\cX}_i \setminus \overline{\Psi} = \overline{\cX}_i$. Then we get
    \begin{equation}
    \prod_{(i, j) \in (\overline{\cX}_i \setminus \overline{\Psi})} \Pr(g_{ij}=0) = \prod_{(i, j) \in \overline{\cX}_i} \Pr(g_{ij}=0) = \Pr(y_i=0).
    \end{equation}
\item There exists $j \in \supp(\bX)$ such that $g_{ij} = 1$ and $(i, j) \not\in \overline{\Psi}$. It is straightforward that $(i, j) \in \overline{\cX}_i \setminus \overline{\Psi}$. Then $\Pr(g_{ij} = 0) = 0$. This implies $\prod_{(i, j) \in(\overline{\cX}_i \setminus \overline{\Psi})} \Pr(g_{ij}=0) = 0$. On the other hand, because $y_i = \bigvee_{(i, j) \in \overline{\cX}_i} g_{ij}$, we get $y_i = 1$ then $\Pr(y_i = 0) = 0$. Eq.~\cref{eqn:not_info} thus holds.
\end{itemize}

Now, we are ready to prove the theorem. Consider $\Psi_a \in \Psi$, since all the missing entries are generated independently, we get
    \begin{equation}
        I(\Psi; y_i) = \sum_{\Psi_a \in \Psi} I(\Psi_a;y_i).
    \end{equation}

Therefore, if $I(\Psi_a; y_i)=0$ for all $\Psi_a \in \Psi$, then $I(\Psi; y_i)=0$. Indeed, we have:
    \begin{align}
        I(\Psi_a; y_i) &= \displaystyle\sum_{\alpha \in \{0,1\}} \displaystyle\sum_{\beta \in \{0,1\}} \Pr(\psi_a=\alpha,y_i=\beta) \log{\frac{\Pr(y_i=\beta|\psi_a=\alpha)}{\Pr(y_i=\beta)} }. \label{eqn:conciseI}
    \end{align}

By combining~\cref{eqn:not_info} and the fact that every entry in $\mM$ is generated independently we have:
    \begin{align*}
        \Pr(y_i|\psi_a) = \prod_{(i, j) \in (\overline{\cX}_i\setminus \overline{\Psi})} \Pr(g_{ij}=0| \psi_a) = \prod_{(i, j) \in (\overline{\cX}_i \setminus \overline{\Psi})} \Pr(g_{ij}=0) = P(y_i).
    \end{align*}

This makes~\cref{eqn:conciseI} become
\begin{align}
        I(\Psi_a; y_i) = \displaystyle\sum_{\alpha \in \{0,1\}} \displaystyle\sum_{\beta \in \{0,1\}} \Pr(\psi_a=\alpha,y_i=\beta) \log{1} = 0.
    \end{align}
\par This completes our proof.
\end{proof}

For a given $\bX \in \operatorname{col}(\mathbf{X})$, when finding the vector induced by missing entries based on the converted missing matrix$\Gamma$ and the converted missing vector $\bV$, we capture all the information about all the pairs $(\bX,i)$ that are informative for $i\in [t]$. In particular, when $(\bX, i)$ is informative, we get:

\begin{equation}
    y_i = \bigvee_{(i, j) \in \overline{\cX}_i \cap \overline{\Psi} } g_{ij}.
\end{equation}

\par It should be noted that the matrix $\Gamma$ and the vector $\mathbf{v}$ are treated as additional information in our matrix completion problem. In scenarios where $\Gamma$ and $\bV$ are insufficient to recover $\Psi$, additional constraints must be imposed on $\mM$ to enable full recovery. We emphasize that there are many such scenarios; in these cases, the number of informative pairs may be insufficient, or there may be no informative pairs at all. An example of such a case is presented in Claim \ref{claim infe} in Appendix~\ref{app:alg:infeasible}.

\subsection{The computational complexity}

In this section, we prove that the MCGT problem is NP-complete as follows.

\begin{theorem}
 The MCGT problem is NP-complete.
\label{thm:complexity}
\end{theorem}

\begin{proof}
    Consider the MCGT problem as described in Section~\ref{sec:intro}. Let $\Gamma$ and $\bV$ be the output of $\mathrm{missingBMF2GT}$($\mG,\mX,\mY,t)$. Due to Corollary~\ref{cor:GammaV}, we have $\Gamma \odot \bPsi = \bV$. On the other hand, because of Theorem~\ref{thm:mutual_info}, finding missing entries in the missing matrix $\mG$ is equivalent to finding $\bPsi$ given $\Gamma$ and $\bV$. Therefore, the hardness of the MCGT problem is equivalent to the hardness of finding $\bPsi$ given $\Gamma$ and $\bV$. Since the converted missing matrix $\Gamma$ is generated randomly, it can considered as a general measurement matrix over the sample space of measurement matrices. Therefore, by using Theorem~\cite[Theorem 6.3.1]{du2000combinatorial} (with $k = 1$), recovering $\bPsi$ given $\Gamma$ and $\bV$ is an NP-complete problem. Thus the MCGT problem is also an NP-complete problem.
\end{proof}

\section{On number of rows in converted missing matrices}
\begin{figure}[!ht]
    \centering
    \scalebox{0.9}{
    \begin{tikzpicture}[
        node distance=1cm and 1cm, 
        every node/.style={font=\small, align=center},
        box/.style={rectangle, draw, rounded corners=6pt, minimum width=4cm, minimum height=1cm, thick},
        theorem/.style={box, fill=red!20},
        lemma/.style={box, fill=blue!15},
        arrow/.style={->, very thick, darkgray} 
    ]
    
    \node (lemma1) [lemma] {\scalebox{0.8}{\textbf{Correlation between total}}\\ \scalebox{0.8}{\textbf{and per-test informative pairs}} \\\\ ~\cref{lem:oneRow}};
    \node (theorem3) [theorem, right=of lemma1] {\scalebox{0.8}{\textbf{Expected informative pair}}\\\scalebox{0.8}{ \textbf{per test with one sample} }\\\\ \cref{thm:s=1}};
    \node (lemma2) [lemma, right=of theorem3] {\scalebox{0.8}{\textbf{Raw calculations for}}\\ \scalebox{0.8}{\textbf{expected informative pair per test} }\\\\ \cref{general_E_for}};
    \node (theorem4) [theorem, right=of lemma2] {\scalebox{0.8}{\textbf{Formula for the expectation}}\\ \scalebox{0.8}{\textbf{of the exact number of rows}} \\\\ \cref{theo5}};
    
    \node (lemma3) [lemma, above=of lemma2, yshift=0.3cm] {\scalebox{0.8}{\textbf{Probability of 2 samples}}\\ \scalebox{0.8}{\textbf{being identical with respect}}\\ \scalebox{0.8}{\textbf{to a given test} }\\\\ \cref{same_pair_lemma}};
    \node (lemma4) [lemma, below=of lemma2, yshift=-0.3cm] {\scalebox{0.8}{\textbf{Calculations of the phi}}\\ \scalebox{0.8}{\textbf{function given in Lemma 2}} \\\\ \cref{phi_cal}};

    \draw[arrow] (lemma1.east) -- (theorem3.west) node[midway, above] {(1)};
    \draw[arrow] (theorem3.east) -- (lemma2.west) node[midway, above] {(2)};
    \draw[arrow] (lemma2.east) -- (theorem4.west) node[midway, above] {(3)};
    \draw[arrow] (lemma3.south) -- (theorem4.north) node[midway, above] {(4)};
    \draw[arrow] (lemma4.north) -- (theorem4.south) node[midway, above] {(5)};

    \end{tikzpicture}}
    \caption{The process of proving the formula for the exact number of rows in theconverted missing matrix}
    \label{fig:Exact row proc}
\end{figure}

    \par It is well known in group testing that the more rows a measurement matrix has, the better the chance we have of recovering the defective items. Therefore, our goal is to approximate the number of rows in the converted missing matrix $\Gamma$. We begin by simplifying the problem through a connection between the total and per-test informative pairs, as established in \cref{lem:oneRow}. This connection allows us to split the expectations and derive \cref{thm:s=1}. Next, we apply the inclusion-exclusion principle to extend \cref{thm:s=1} to the more general case in \cref{general_E_for}. Finally, by incorporating results from \cref{same_pair_lemma} and \cref{phi_cal}, we derive the exact expression for the expectation, as stated in \cref{theo5}. This overall process is illustrated in \cref{fig:Exact row proc}.

\subsection{On exact number of rows in missing matrices}
\label{sec:estimation}

Since it is redundant to have duplicated rows in $\Gamma$ (with the same test outcomes) in Theorem~\ref{thm:construction}, we define the concepts of different and identical informative pairs as follows.

\begin{definition} [Identical informative pairs]
    For an informative pair $(\bX, i)$, let $\Psi_{(\bX,i)} \subseteq \overline{\Psi}$ be the set of all missing entries lying on row $i$ of $\mM$ such that for all $(i, j) \in \Psi_{(\bX,i)}$, $x_j=1$. Then, for two informative pairs $(\bX, i_1)$ and $(\bX^\prime, i_2)$, they are identical if and only if $\Psi_{(\bX, i_1)} \equiv \Psi_{(\bX^\prime, i_2)}$.
    \label{dif-same-def}
\end{definition}

\par From this definition, two informative pairs $(\bX, i_1)$ and $(\bX^\prime, i_2)$ are different if and only if $\Psi_{(\bX, i_1)} \not\equiv \Psi_{(\bX^\prime, i_2)}$. This can be interpreted as follows: in the process of constructing the converted missing matrix $\Gamma$, these two pairs create two different rows, i.e., two rows that differ in at least one column. This definition is later used to estimate the number of rows in $\Gamma$. It should also be noted that if $i_1 \neq i_2$, then the pairs $(\bX, i_1)$ and $(\bX', i_2)$ will always be different for all $\bX$ and $\bX'$, since their missing entries sets belong to different rows and thus do not intersect.

 Given $\bX, \bY$, and $\bG$, we aim to find the expected value of the number of rows in $\Gamma$. Since the entries in $\mG$ are independently and identically distributed, we can utilize the linearity property of the expected value as follows.

\begin{lemma}
Let $t, \mathbf{X}, \Psi$ be variables defined in~\cref{sec:intro}. Let $\cH \subseteq \operatorname{col}(\mathbf{X})$ be the set of pairwise different informative pairs generated by all tests and $\mathbf{X}$, i.e., $\forall \bX \neq \bX^\prime \in \cH$ and $\forall i \neq i^\prime \in [t]$, $\Psi_{(\bX, i)} \not\equiv \Psi_{(\bX^\prime, i^\prime)}$. For $i \in [t]$, let $\mathcal{X}_i \subseteq \operatorname{col}(\mathbf{X})$ be the set of pairwise different informative pairs generated by row $i$ and columns in $\operatorname{col}(\mathbf{X})$, i.e., $\forall \bX \neq \bX^\prime \in \mathcal{X}_i$, $\Psi_{(\bX, i)} \not\equiv \Psi_{(\bX^\prime, i)}$. Set $h := |\cH|$ and $\eta_i = |\mathcal{X}_i|$. For any $\tau \in [t]$, we have
\begin{equation}
    \bbE[h] = t \bbE[\eta_\tau].
\end{equation}
\label{lem:oneRow}
\end{lemma}

\begin{proof}
    We have:
    \begin{align*}
        \bbE[h] = \bbE\left[\displaystyle\sum_{\tau \in [t]} \eta_\tau\right] = \displaystyle\sum_{\tau \in t} \bbE[\eta_\tau] = t \bbE[\eta_\tau].
    \end{align*}
    \par The first equality comes from the fact that for any $i \neq j \in [t]$ and for any $\bX, \bY \in \operatorname{col}(\mathbf{X})$, if both $(\bX, i)$ and $(\bY,j)$ are informative, they are different. The second and third equations are due to each entry is independent and identically distributed.
\end{proof}

\par Because of~\cref{lem:oneRow}, instead of estimating $\bbE[h | s]$ by dealing with the whole matrix $\mM$, we will only work with one row, denoted as $\tau$, in $\mM$. For consistency and easy understanding, we replace $\eta_\tau$ by $\omega$ and our task is to estimate $\bbE[\omega]$.

\subsubsection{On \texorpdfstring{$s=1$}{}}

When $s = 1$, $\bbE[\omega]$ can be calculated as follows.

\begin{lemma}
    Let $p,q, d,s, \mathbf{X}$ be variables that have been defined in \cref{sec:intro}. Suppose $s = 1$ and $\bX \in \operatorname{col}(\mathbf{X})$. We have:
    \begin{equation}        
    \bbE[\omega = 1] = \Pr((\bX,\tau) \mbox{ is informative}) = (1-p+pq)^d-\left[(1-p)(1-q)\right]^d.
    \end{equation}
    \label{thm:s=1}
\end{lemma}
\begin{proof}
    Since $|\operatorname{col}(\mathbf{X})| = 1$ and $\bX \in \operatorname{col}(\mathbf{X})$, we have:
    \begin{align*}
        \bbE[\omega]&= 0 \cdot \Pr(\omega = 0 | s) + 1\cdot \Pr(\omega = 1 | s) = \Pr((\bX,\tau) \text{ is informative}) \\
        &=(1-p+pq)^d-\left[(1-p)(1-q)\right]^d.
    \end{align*}
    \par The first part is to satisfy the second condition of~\cref{informative-pair}. The second part is to remove the case when all entries at row $\tau$ are zeros, and therefore the first condition of~\cref{informative-pair} holds.
\end{proof}

\subsubsection{On \texorpdfstring{$s \geq 1$}{s 1}}

First, we derive a formal expression for the probability that a subset of $\operatorname{col}(\mathbf{X})$ contains elements that are all informative and pairwise identical with respect to a random generated row $\tau \in [t]$.
\begin{definition} \label{def theta_=}
For all $\theta=\{\bX_{\alpha_1},\bX_{\alpha_2},\dots,\bX_{\alpha_z}\}\subseteq\operatorname{col}(\mathbf{X})$ (where $\mathbf{X}$ is defined in \cref{sec:intro}), we denote $\Pr(\bX_{\alpha_1}=\bX_{\alpha_2}=\dots=\bX_{\alpha_z})=\Pr(\theta_=)$ the probability that $( \bX_{\alpha_1}, \mathbf{\tau}),\dots,(\bX_{\alpha_z}, \mathbf{\tau})$ are all informative and are pairwise identical, with respect to the random generating of row $\tau$ of $\mathbf{M}$ following \cref{ran_prob}.
\end{definition}

Similar to~\cref{def theta_=}, the following definition provides a formal expression for the expected number of subsets of $\operatorname{col}(\mathbf{X})$ in which every element is both informative and pairwise identical with respect to a random generated row $\tau \in [t]$, considering all possible subsets of $\operatorname{col}(\mathbf{X})$. 
\begin{definition} \label{def expect chi_z}
For any positive integer \( l \) and \( \mathbf{X} \), we define \(\bbE[\mathbf{X}_l]\) as the expected number of \( l \)-element subsets of \( \operatorname{col}(\mathbf{X}) \), denoted by \( (\bX_{\beta_1}, \dots, \bX_{\beta_l}) \), such that the tuples \( (\bX_{\beta_1}, \tau), \dots, (\bX_{\beta_l}, \tau) \) are both informative and pairwise identical.
\end{definition}

Then, the exact value of~\eqref{lem:oneRow} dividing by the number of rows in the measurement matrix $\mM$ is summarized in the following theorem and is proved in Appendix \ref{append: A}.

\begin{theorem}\label{theo5}
Let $t, n, d, p, q$, $\mathbf{X}$ and $\cT_d$ be defined in~\cref{sec:intro}. Let $\cH \subseteq \operatorname{col}(\mathbf{X})$ be the set of pairwise different informative pairs generated by all tests and $\mathbf{X}$, i.e., $\forall \bX \neq \bX^\prime \in \cH$ and $\forall i \neq i^\prime \in [t]$, $\Psi_{(\bX, i)} \not\equiv \Psi_{(\bX^\prime, i^\prime)}$. For any $\tau \in [t]$, let $\mathcal{X}_\tau \subseteq \operatorname{col}(\mathbf{X})$ be the set of pairwise different informative pairs generated by row $\tau$ and columns in $\operatorname{col}(\mathbf{X})$, i.e., $\forall \bX \neq \bX^\prime \in \mathcal{X}_\tau$, $\Psi_{(\bX, \tau)} \not\equiv \Psi_{(\bX^\prime, \tau)}$. Set $h:= |\cH|$ and $\omega := |\mathcal{X}_\tau|$. Then
    \begin{equation}    \label{theo5 eq}    
        \dfrac{\bbE[h]}{t}=\bbE[\omega]=s\Upsilon(d)-\dfrac{\binom{n}{d}\binom{\binom{n}{d}-2}{s-2}}{2\binom{\binom{n}{d}}{s}}\cdot \left[\displaystyle\sum_{i=0}^{d-1}\binom{d}{i}\binom{n-d}{d-i}\Upsilon(i)\right]+\displaystyle\sum_{c=3}^{s}(-1)^{c+1}\dfrac{\displaystyle\sum_{\operatorname{col}(\mathbf{X})\subseteq \cT_d}\displaystyle\sum_{\theta\subseteq\operatorname{col}(\mathbf{X}),|\theta|=c}\Pr(\theta_=)}{\binom{\binom{n}{d}}{s}},    
    \end{equation}
where
\begin{align}\label{Ups_equ}
    \Upsilon(u)=\left[(1-p)(1-q)\right]^{2d-2u}(1-p+pq)^{u}-\left[(1-p)(1-q)\right]^{2d-u},
\end{align}
and $\Pr(\theta_=)$ is the probability that for every $\mathbf{\alpha}, \mathbf{\beta} \in \theta \subseteq \operatorname{col}(\mathbf{X})$, two informative pairs $(\mathbf{\alpha}, \tau)$ and $(\mathbf{\beta}, \tau)$ are identical.
\label{thm:Formula_E}
\end{theorem}
\par Additionally, to better characterize the range of $\Upsilon(d)$, we derive bounds for this quantity in \cref{bound for Ups} in Appendix \ref{append: A}.

\subsection{On an approximate number of rows in converted missing matrices}
\label{sec:approximate}

\begin{figure}[!ht]
    \centering
    \begin{tikzpicture}[
            node distance=1cm and 1cm, 
            every node/.style={font=\small, align=center},
            box/.style={rectangle, draw, rounded corners=6pt, minimum width=4cm, minimum height=1cm, thick},
            theorem/.style={box, fill=red!20},
            lemma/.style={box, fill=blue!15},
            arrow/.style={->, very thick, darkgray} 
        ]
    
        \node[lemma] (L6) {\scalebox{0.8}{\textbf{Increasing property of}} \\ \scalebox{0.8}{\textbf{binomial product with sufficiently large}} $n$ \\\\ \cref{binom inq lemma}};
        \node[lemma, below=of L6] (L7) {\scalebox{0.8}{\textbf{Increasing property of}} \\ \scalebox{0.8}{\textbf{Upsilon function}} \\\\ \cref{Ups inc lemma}};
        \node[lemma, right=of L6] (L8) {\scalebox{0.8}{\textbf{Upper bound for Omega term}} \\ \scalebox{0.8}{\textbf{given in the exact formula of}} \\ \scalebox{0.8}{\textbf{rows of the converted missing matrix}} \\\\ \cref{lem Omega}};
        \node[lemma, below=of L8] (L5) {\scalebox{0.8}{\textbf{Inclusion-Exclusion type bound for}} \\ \scalebox{0.8}{\textbf{the number of rows in the converted missing matrix}} \\ \scalebox{0.8}{\textbf{with complex Omega term}} \\\\ \cref{Inclusion-Exclusion Bound}};
        \node[theorem, right=of L8] (T3) {\scalebox{0.8}{\textbf{Closed-form bound for the number}} \\ \scalebox{0.8}{\textbf{of rows in the converted missing matrix}} \\\\ \cref{Expectation_bound}};
    
        \draw[arrow] (L6.east) -- (L8.west) node[midway, above] {(1)};
        \draw[arrow] (L7.east) -- (L8.south west) node[midway, above] {(2)};
        \draw[arrow] (L8.east) -- (T3.west) node[midway, above] {(3)};
        \draw[arrow] (L5.north east) -- (T3.south west) node[midway, above] {(4)};
    
    \end{tikzpicture}
    \caption{The process of proving an approximate bound for the number of row in the converted missing matrix.  }
    \label{fig:Upper Bound row proc}
\end{figure}
\par \cref{theo5} provides an exact calculation for $\bbE[h]$. However, the final component of this formula is quite complex, making its practical computation unrealistic in real-world scenarios. To address this, this section will instead provide bounds for $\bbE[h]$ using much simpler formulas. We begin by deriving two lemmas: one that establishes the monotonicity of the Upsilon function (\cref{Ups inc lemma}), and another that presents a special inequality involving a product of binomial coefficients (\cref{binom inq lemma}). These two results are then used to bound the most complex term in \cref{theo5}, as shown in \cref{lem Omega}. Finally, we apply the Inclusion-Exclusion principle to derive an additional bound in \cref{Inclusion-Exclusion Bound}. Together with the result from \cref{lem Omega}, this allows us to obtain an upper bound on the number of rows in the converted missing matrix, as presented in \cref{Expectation_bound}. This process is shown in \cref{fig:Upper Bound row proc}. As stated, the main result of this section is \cref{Expectation_bound}, which is presented below. The proof is deferred to Appendix \ref{append: B}.

\begin{theorem}  \label{Expectation_bound}
Let $t, n, d,s$ be defined in~\cref{sec:intro}, $\Upsilon$ be defined in \cref{Ups_equ}, and $h$ be the number of rows of the converted missing matrix. Then, we have:
    \begin{align}\label{theo6 eq1}
    s\Upsilon(d)\left\{1-\dfrac{s-1}{2}\cdot        
     \dfrac{\displaystyle\sum_{i=0}^{d-1}\binom{d}{i}\binom{n-d}{d-i}\Upsilon(i)}{\displaystyle\sum_{i=0}^{d-1}\binom{d}{i}\binom{n-d}{d-i}\Upsilon(d)}\right\}\leq \dfrac{\bbE[h]}{t} \leq s\Upsilon(d).
\end{align}
Furthermore, when $n>(d+1)^2$, we have:
\begin{align}\label{sample-bound}
    s\Upsilon(d)\left[1-\dfrac{s-1}{2d}\cdot        
     \dfrac{a^2}{b-a^2}\right]
\leq \dfrac{\bbE[h ]}{t} \leq s\Upsilon(d),
\end{align}
where:
\begin{align*}
    a=(1-p)(1-q), b=1-p+pq, \Upsilon(d)=b^d-a^d.
\end{align*}
\end{theorem}

\subsection{On the concentration of \texorpdfstring{$\bbE[h]$}{\bbE[h]}}

In this section, we will be focusing on seeing how close the approximation of the expected value of rows of the converted missing matrix in \cref{sec:approximate} to its true value. Specifically, for $\delta>0$, we will bound the following term
$$\Pr(|h-\mathbb{E}[h]|>\delta).$$

For a given $\mathbf{X}=\left[\mathbf{x}_1,\dots,\mathbf{x}_s\right]$, recall the definition of $\eta_i\quad i\in[t]$ in \cref{lem:oneRow}, we have
\begin{align*}
    \mathbb{E}[h^2]=\mathbb{E}\left[\left(\displaystyle\sum_{i=1}^t \eta_i\right)^2\right]=\mathbb{E}\left[\displaystyle\sum_{i=1}^t \eta_i^2\right]+\mathbb{E}\left[\displaystyle\sum_{\substack{i\neq j\\i,j\in[t]}}\eta_i\eta_j\right]
\end{align*}
But since $\eta_i$ and $\eta_j$ are identically independent distributed, this leads to  
\begin{align}\label{h2 split}
    \mathbb{E}[h^2]=t\mathbb{E}\left[ \tau^2\right]+t(t-1)\mathbb{E}\left[\tau\right]^2,
\end{align}
where $\omega$ is defined similar to that in \cref{thm:s=1}. Moreover, from the same theorem, $\bbE[\tau]=s\Upsilon(d)$. Thus, our target is to calculate $\bbE[\tau^2]$. For $j \in [s]$, we denote the random variables $\mathbf{y}_j$ that satisfy $\mathbf{y}_j=1$ iff $(\mathbf{x}_j,\tau)$ is informative. Hence,
\begin{align}\label{equa: E[h^2] bound}
    \bbE[\tau^2]\overset{(1)}{\leq}\bbE\left[\left(\sum_{j=1}^s\mathbf{y}_j\right)^2\right]\overset{(2)}{=}s\bbE[\mathbf{y_1}^2]+\sum_{\substack{i\neq j\\i,j\in[s]}}\bbE[\mathbf{y}_i\mathbf{y}_j]\overset{(3)}{=}s\Pr(\mathbf{y}_1=1)+\sum_{\substack{i\neq j\\i,j\in[s]}}\Pr(\mathbf{y}_i=\mathbf{y}_j=1)
\end{align}

Inequality (1) in \cref{equa: E[h^2] bound} arises because informative pairs can create duplicate rows in the converted missing matrix, thereby reducing the number of distinct rows. The subsequent equality (2) hold because $\mathbf{y}_i$ and $\mathbf{y}_j$ are iid\ for all $i \neq j$. The final equality (3) follows by the same argument as in \cref{thm:s=1}. Now, by calculating the probabilities similar to \cref{thm:s=1}, we get
\begin{align*}
    \Pr(\mathbf{y}_1=1)&=\Upsilon(d),\\
    \Pr(\mathbf{y}_i=\mathbf{y}_j=1)&=(1-p+pq)^{2d-\simm(\mathbf{x}_1,\mathbf{x}_2)}+\left[(1-p)(1-q)\right]^{2d-\simm(\mathbf{x}_1,\mathbf{x}_2)}-2\left[(1-p)(1-q)\right]^{d},
\end{align*}
where we recall that $\simm(\mathbf{x}_1,\mathbf{x}_2)=\bX_1^{\top}\bX_2$. Thus, by combining the above with \cref{h2 split} and \cref{equa: E[h^2] bound}, we yield
\begin{align}\label{bound h^2}
     \mathbb{E}[h^2]\leq ts\Upsilon(d) +t(t-1)s^2\Upsilon(d)^2+t\sum_{\substack{i\neq j\\i,j\in[s]}}\mathcal{J}(i,j)
\end{align}
where $\mathcal{J}(i,j)=(1-p+pq)^{2d-\simm(\mathbf{x}_1,\mathbf{x}_2)}+\left[(1-p)(1-q)\right]^{2d-\simm(\mathbf{x}_1,\mathbf{x}_2)}-2\left[(1-p)(1-q)\right]^{d}$
Furthermore, by applying Chebyshev's inequality for concentration, one can show that
$$ \Pr(|h-\bbE[h]|>\delta)\leq \dfrac{Var(h)}{\delta^2}=\dfrac{\bbE[h^2]-\bbE[h]^2}{\delta^2},$$
and thus, combining with \cref{bound h^2} and \cref{theo6 eq1}, for a pre-given $\mathbf{X}=\left[\mathbf{x}_1^{\top},\dots,\mathbf{x}_s^{\top}\right]$, we have
\begin{align}\label{eq: pregiven bound}
        \Pr(|h-\bbE[h]|>\delta)\leq \dfrac{st\Upsilon(d)+t(t-1)s^2\Upsilon(d)^2+t\displaystyle\sum_{\substack{i\neq j\\i,j\in[s]}}\mathcal{J}(i,j)-s^2t^2\Upsilon(d)^2\mathcal{W}(a,b)^2}{\delta^2},
    \end{align}

where we recall that $ \simm(\mathbf{x}_i,\mathbf{x}_j)=\mathbf{x}_i^{\top}\mathbf{x}_j$ and $\mathcal{W}(a,b) = 1-\dfrac{s-1}{2d} \cdot \dfrac{a^2}{b-a^2}$.\\
In order to generalize \cref{eq: pregiven bound} to the case where $\mathbf{X}$ satisfies $\operatorname{col}(\mathbf{X})$ is a random subset of $\mathcal{T}_d$, the set of all binary vectors with weight $d$. We notice that  $\simm(\mathbf{x}_i,\mathbf{x}_j)\leq d$ for all $\mathbf{x}_i,\mathbf{x}_j$. Thus
\begin{align*}
    \mathcal{J}(i,j)\leq \Upsilon(d)\quad \text{for all }i,j
\end{align*}
Combining this with \cref{eq: pregiven bound}, we yield the general bound for the concentration of $h$ as follow
\begin{theorem} \label{theo concentration bound}
    Let $t, n, d,s$ be defined in~\cref{sec:intro}, a and b be defined in \cref{Expectation_bound}, $\Upsilon$ be defined in \cref{Ups_equ}. Additionally, given  where each element in $\operatorname{col}(\mathbf{X})$ is taken uniformly from $\mathcal{T}_d$ and assume $n>(d+1)^2$, for every $\delta>0$, we have 
    \begin{align*}
        \Pr(|h-\bbE[h]|>\delta)\leq \dfrac{t(t-1)s^2\Upsilon(d)^2+ts^2\Upsilon(d)-s^2t^2\Upsilon(d)^2\mathcal{W}(a,b)^2}{\delta^2},
    \end{align*}
    where
    \begin{align*}
        \simm(\mathbf{x}_i,\mathbf{x}_j)&=\mathbf{x}_i^{\top}\mathbf{x}_j\\
        \Upsilon(u) &= \left[(1-p)(1-q)\right]^{2d-2u}(1-p+pq)^{u}-\left[(1-p)(1-q)\right]^{2d-u}\\
        \mathcal{W}(a,b) &= 1-\dfrac{s-1}{2d} \cdot \dfrac{a^2}{b-a^2}
    \end{align*}
\end{theorem}
\begin{remark}
    It is worth noting that the numerator of the LHS bound in \cref{theo concentration bound} can be rewritten as
    \begin{align*}
        ts^2\Upsilon(d)[1-\Upsilon(d)]+t^2s^2\Upsilon(d)^2\left[1-\mathcal{W}(a,b)^2\right]
    \end{align*}
    Thus, by applying the framework of the Central Limit Theorem, one can conclude that
    \begin{enumerate}
        \item When $d$ and $n$ goes to infinity (such as the well-known case where $d=n^\alpha,\alpha>0$ and $n\rightarrow\infty$) then $\Pr(|h-\bbE[h]|>\delta)\rightarrow0\quad\forall\delta>0$. Hence, $h$ converges in distribution to a normal distribution with mean $\bbE[h]$.
        \item This property no longer holds when we consider \( s \to \infty \) instead of \( d \). As \( s \) increases, the number of rows in \( \Gamma \) grows much more rapidly, which gives \( h \) a greater chance of significantly exceeding \( \mathbb{E}[h] \).
    \end{enumerate}
    Nevertheless, even though increasing \( s \) does not guarantee that \( h \) will be arbitrarily close to its expected value, our work shows that doing so increases the success rate of solving the overall matrix completion problem. In most cases, this success rate approaches one. This result is presented in \cref{theo: universal bound for GT} and \cref{remark: go to 1}.
\end{remark}

\section{Bounds for success recovery probability}
\label{sec: Gen bound for SP}

In this section, we present a general bound on the success probability of solving the converted missing matrix, given the initial matrix \( \mM \). We begin by reformulating the definition of the converted missing matrix without altering its properties or the information it provides for the matrix completion problem. This reformulation is presented in \cref{alter def gamma}. Using this new definition, we derive our main result, \cref{theo: universal bound for GT}. To prove \cref{theo: universal bound for GT}, we show \cref{independent of success}, which decomposes the process of solving \( \Gamma \) into multiple group testing sub-problems. Furthermore, we demonstrate in \cref{ber gen} that these smaller group testing problems are also Bernoulli-generated. By combining these insights, we complete the proof for Proposition \ref{theo: general_bound}. Finally, we apply this bound to classical algorithms for the group testing problem. This process is shown in \cref{fig:Upper Bound suc prob proc}. Additionally, after we have derived Proposition \ref{theo: general_bound}, we combine it with some relaxations to get \cref{theo: universal bound for GT}, which is a universal bound for the success probability of solving the group testing problem via the recovered matrix.

\begin{figure}
    \centering
\scalebox{0.775}{\begin{tikzpicture}[
        node distance=1cm and 1cm, 
        every node/.style={font=\small, align=center},
        box/.style={rectangle, draw, rounded corners=6pt, minimum width=3.5cm, minimum height=1cm, thick},
        theorem/.style={box, fill=red!20},
        lemma/.style={box, fill=blue!15},
        definition/.style={box, fill=blue!10},
        result/.style={draw, ellipse, fill=green!20, minimum width=4.5cm, minimum height=1.5cm, thick, scale=0.8},
        arrow/.style={->, very thick, darkgray} 
    ]

    \node[definition] (D5) {\scalebox{0.8}{\textbf{Alternative definition of}} \\ \scalebox{0.8}{\textbf{the converted missing matrix and the converted missing vector}} \\\\ \cref{alter def gamma}};
    \node[theorem, right=of D5] (T6) {\scalebox{0.8}{\textbf{Splitting the overall success}} \\ \scalebox{0.8}{\textbf{probability into success}} \\ \scalebox{0.8}{\textbf{probabilities of sub-group}} \\ \scalebox{0.8}{\textbf{testing problems}} \\\\ \cref{independent of success}};
    \node[theorem, right=of T6] (T7) {\scalebox{0.8}{\textbf{General bound for the}} \\ \scalebox{0.8}{\textbf{success probability of the}} \\ \scalebox{0.8}{\textbf{matrix recovery problem} }\\\\ Proposition \ref{theo: general_bound}};
    \node[lemma, below=of T6] (L9) {\scalebox{0.8}{\textbf{Showing that each sub-group}} \\ \scalebox{0.8}{\textbf{testing problem is}} \\ \scalebox{0.8}{\textbf{Bernoulli-generated}} \\\\ \cref{ber gen}};
    \node[theorem, below=of T7] (L10) {\scalebox{0.8}{\textbf{Success rate of solving GT}} \\ \scalebox{0.8}{\textbf{problem via the}} \\ \scalebox{0.8}{\textbf{recovered matrix}} \\\\ \cref{theo: universal bound for GT}};

    \node[result, right=of T7, yshift=2.5cm] (R1) {\textbf{Bound for }success probability\textbf{ of the} \\ \textbf{matrix recovery problem using} \\ \textbf{COMP algorithm}};
    \node[result, right=of T7] (R2) {\textbf{Bound for success probability of the} \\ \textbf{matrix recovery problem using} \\ \textbf{DD algorithm}};
    \node[result, right=of T7, yshift=-2.5cm] (R3) {\textbf{Bound for success probability of the} \\ \textbf{matrix recovery problem using} \\ \textbf{SSS algorithm}};

    \draw[arrow] (D5.east) -- (T6.west) node[midway, above] {(1)};
    \draw[arrow] (T6.east) -- (T7.west) node[midway, above] {(2)};
    \draw[arrow] (L9.north) -- (T7.south west) node[midway, below] {(3)};
    \draw[arrow] (T7.east) -- (R1.west) node[midway, left] {(4)};
    \draw[arrow] (T7.east) -- (R2.west) node[midway, above] {(5)};
    \draw[arrow] (T7.east) -- (R3.west) node[midway, right] {(6)};
    \draw[arrow] (T7.south) -- (L10.north) node[midway, right] {(6)};

\end{tikzpicture}}
    \caption{The process of proving bounds for the success probability.}
    \label{fig:Upper Bound suc prob proc}
\end{figure}

\subsection{General bound}
\par First, for the purpose of convenience, we adjust how we define the converted missing matrix $\Gamma$ and the converted missing vector $\bV$ so that their number of rows is constant while still being equivalent to the previous definition. For clarity, we denote the newly defined $\Gamma$ and $\mathbf{v}$ by $\tilde{\Gamma}$ and $\tilde{\bV}$, respectively.

\begin{algorithm}
\caption{$\mathrm{missingBMF2GTPlus}(\mM, \mX, \mY, t)$}
\label{alter def gamma}

\begin{algorithmic}[1]
\State \textbf{Input:} Measurement matrix $\mM$, sample matrix $\mathbf{X}$, outcome matrix $\mY$, number of tests $t$.
\State \textbf{Output:} Missing matrix $\tilde{\Gamma}$, converted missing vector $\tilde{\mathbf{v}}$.
\State We begin with the first row in the matrix $\mM$ and the first sample $x_1$ in $\operatorname{col}(\mathbf{X})$.
\State If $(x_1,1)$ is informative, we add a row to $\tilde{\Gamma}$ following the procedure outlined in \cref{thm:construction}.
\State However, if $(x_1,1)$ is not informative, instead of omitting it from $\tilde{\Gamma}$ and $\tilde{\bV}$, we append a row of zeros to $\tilde{\Gamma}$ and a corresponding $0$ to $\tilde{\bV}$. 
\State This process is repeated for all columns in $\operatorname{col}(\mathbf{X})$ and subsequently for all rows in $\mM$.
\end{algorithmic}
\end{algorithm}

\begin{figure}[!ht]
    \centering
    \includegraphics[width=0.8\linewidth]{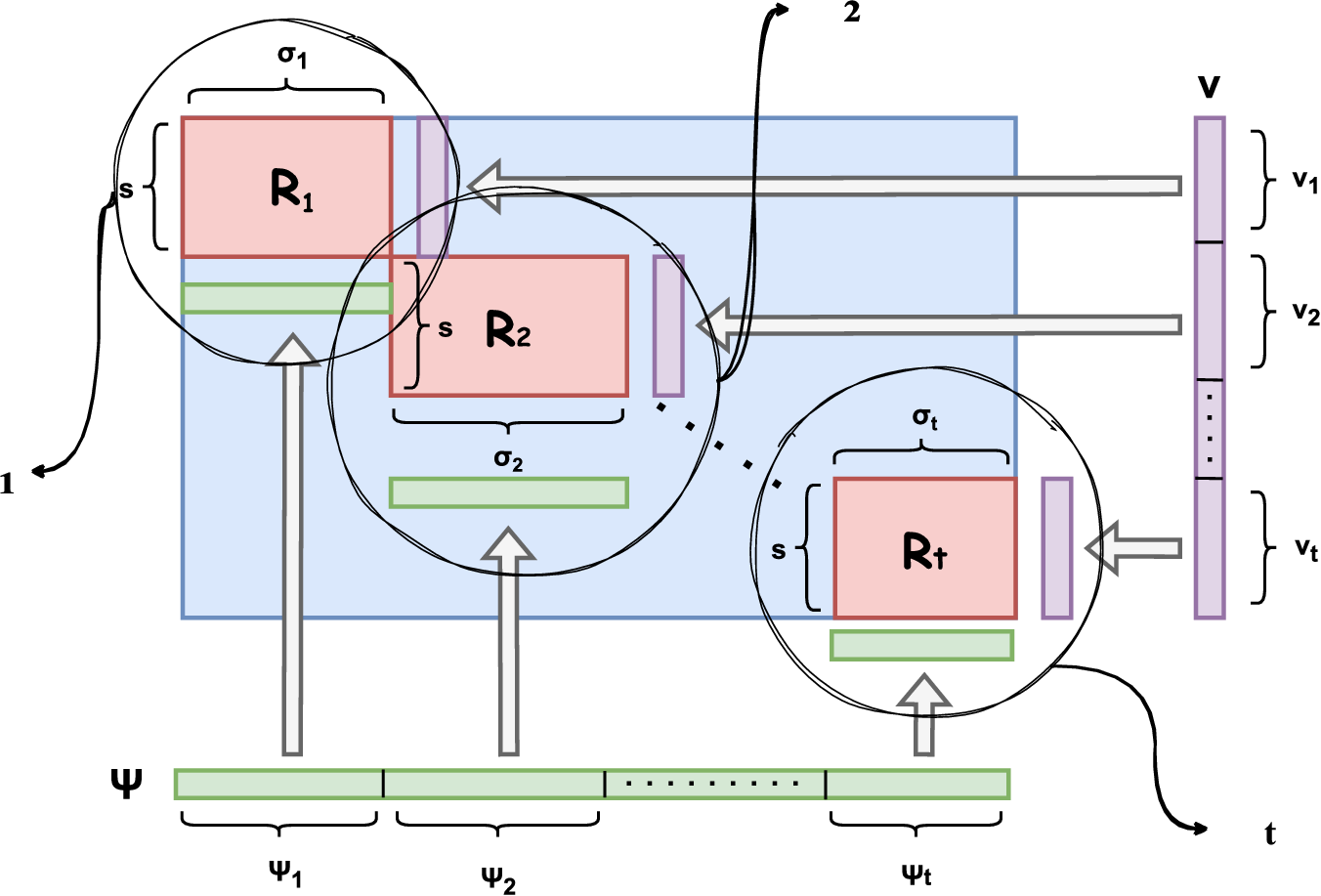}
    \caption{Demonstration of $\tilde{\Gamma}$ and $\tilde{\bV}$ defined via \cref{alter def gamma}.}
    \label{fig:demo for def 5}
\end{figure}
\par The procedure to construct $\tilde{\Gamma}$ and $\tilde{\bV}$ is described in words in \cref{alter def gamma}. A more formal math description of this procedure is shown in terms of pseudocode in \cref{Algo2} and is demonstrated graphically in \cref{fig:demo for def 5}. Upon completion, we obtain a matrix $\tilde{\Gamma}$ and a vector $\tilde{\bV}$ with exactly $s \times t$ rows, which remains equivalent to the previous construction method in \cref{thm:construction} in terms of finding the missing cells using group testing.  Next, using the new construction shown in \cref{alter def gamma}, we decompose $\tilde{\Gamma}$ and $\tilde{\bV}$ into multiple group testing problems $\mathcal{Q}_i$, which are defined in \cref{def: things related to def 5}, as follow
\begin{definition}\label{def: things related to def 5}
For all $i \in [t]$, we denote $\sigma_i$ and $\varphi_i$ to be the number of missing cells and the number of cells with number 1 in row $i$ of $\mG$, respectively. Additionally, we also denote $\mathcal{R}_i$ to be the matrix that is created by the set of row from row $(s(i-1)+1)^{th}$ to row $si^{th}$ and the set of columns from column $(\sigma_1+\dots+\sigma_{i-1}+1)^{th}$ to $(\sigma_1+\dots+\sigma_i)^{th}$ of $\tilde{\Gamma}$ . Similarly, $\tilde{\bV}_i$ is denoted to be a slice of the vector $\tilde{\bV}$ from position $(s(i-1)+1)^{th}$ to $si^{th}$ and $\Psi[i]$ is denoted to be a slice of the vector $\Psi$ from position $(\sigma_1+\dots+\sigma_{i-1}+1)^{th}$ to $(\sigma_1+\dots+\sigma_i)^{th}$. Furthermore, we denote $\bar{\varphi}_i$ to be the number of $\Psi_\alpha\in\Psi[i]$ such that $\psi_\alpha=1$. Lastly, we denote $\mathcal{Q}_i$ the group testing problem containing the measurement matrix $\mathcal{R}_i$, the output vector $\tilde{\bV}_i$ and the sample vector $\Psi[i]$; and $\Pr_{\mathcal{Q}_i}(succ)$ to be the success probability of it (i.e $\Pr_{\mathcal{Q}_i}(succ)=\Pr(\mathcal{Q}_i\text{ is success})$).
\end{definition}

The following lemma presents that the entries in matrix $\cR_\alpha$ for $\alpha \in [t]$ follows Bernoulli distribution and the proof is available in Appendix \ref{proof for ber gen}.

\begin{lemma} \label{ber gen}
    Let $t$ be defined in~\cref{sec:intro} and $\mathcal{R}_\alpha$ be defined in \cref{def: things related to def 5}. For all $\alpha\in [t]$,  $\mathcal{R}_\alpha$ is Bernoulli($\nu_\alpha$) generated, where $\nu_\alpha = \binom{n-\varphi_\alpha-1}{d-1}/{\binom{n}{d}}$.
\end{lemma}

\par It is also worth noting that $\Psi[i]$ can be understood as the number of missing cells in row $i$ of $\mM$. Next, we derive the main result of this section, \cref{theo: universal bound for GT}, which shows the effectiveness of our recovered matrix in terms of the success probability for the group testing problem. It is intuitive that the success probability improves as we obtain more tests (i.e., the number of matrix $\mX$ increases). Suppose we observe a $n\times s$ matrix $\mX$ at \textit{each sampling round}; a natural question is how the accuracy of our method scales with the number of sampling rounds. In the following theorem, we aim to answer this by deriving a universal lower bound on the success probability corresponding to the number of sampling rounds, applicable across different algorithms used to recover the measurement matrix and solve the group testing problem.

\begin{theorem}\label{theo: universal bound for GT}
    Let $\sigma, \varphi,\bar{\varphi}$ be defined in \cref{def: things related to def 5}. Additionally, given two algorithms $\mathcal{A}$ and $\mathcal{B}$, suppose that in a standard group testing setup where there are $n_0$ items, $t_0$ tests, $d_0$ defectives, and the measurement matrix is Bernoulli($p_0$)-generated, $\mathcal{A}$ and $\mathcal{B}$ guarantee a success probability of at least $\Theta(n_0, t_0, d_0, p_0)$ and $\bar{\Theta}(n_0, t_0, d_0, p_0)$, respectively. Consider a group testing problem with $n$ items, $t$ tests, $d$ defectives, and a measurement matrix that is Bernoulli($p$)-generated but has some missing cells. For each sampling round, if we recover the measurement matrix using our method with algorithm $\mathcal{A}$ and the sample matrix $\mX\in\mathbb{R}^{n\times s}$, and then solve the group testing problem via the recovered measurement matrix using algorithm $\mathcal{B}$, the success probability is bounded by number of sampling rounds $\hat{r}$ as follow
    \begin{align*}
        \Pr_{GT}(\text{succ}) \geq 1 - \left( 1 - \left( \min_{(\sigma, \bar{\varphi}, \varphi) \in \mathcal{K}} \Theta\left( \sigma, s, \bar{\varphi}, \nu \right) \right)^t \bar{\Theta}(n, t, d, p) \right)^{\hat{r}},
    \end{align*}
    where $\nu = \binom{n- \varphi - 1}{d - 1}/\binom{n}{d}$, and $\mathcal{K}=\{(\sigma,\bar{\varphi},\varphi\}|\sigma,\bar{\varphi},\varphi\in[n], \bar{\varphi}\leq\sigma, \sigma+\varphi\leq n \}$. 
\end{theorem}

To prove \cref{theo: universal bound for GT}, we will, first, show the correlation between the original problem and a set of smaller group testing problems $\mathcal{Q}_i$ through the next theorem.

\begin{theorem}
\label{independent of success}
    Let $t$ be defined in~\cref{sec:intro} and $\mathcal{Q}_i$ be defined in \cref{def: things related to def 5} for $i\in[t]$, we have:
    \begin{align*}
        \Pr(succ)=\prod_{i=1}^{t} \Pr_{\mathcal{Q}_i}(succ)
    \end{align*}
\end{theorem}

\par The proof is referred to Appendix \ref{proof for theo idsuc}. \cref{independent of success} has given us a way to calculate $\Pr(succ)$ through $\Pr_{\mathcal{Q}_i}$. Our next target would be to bound each $\Pr_{\mathcal{Q}_i}$.By applying \cref{independent of success} and \cref{ber gen}, we derive a general upper bound on the success probability $\Pr(\text{succ})$ of an algorithm $\mathcal{A}$, under the assumption that we have a bound on the success probability of $\mathcal{A}$ in the standard group testing setting where the measurement matrix is generated using a Bernoulli distribution. Additionally, by incorporating the bound from \cref{Expectation_bound}, which estimates the expected number of rows in which at least one element of $\tilde{\Gamma}$ appears, we can further bound the difficult-to-interpret terms in the general expression. This leads to the following proposition, which presents the complete general bound:
\begin{prop} \label{theo: general_bound}
     Let $t, n, d,s$ be defined in~\cref{sec:intro}, a and b be defined in \cref{Expectation_bound}, $\Upsilon$ be defined in \cref{Ups_equ}, $\sigma_i, \varphi_i,\bar{\varphi}_i$ be defined in \cref{def: things related to def 5}. Additionally, given an algorithm $\mathcal{A}$ for the group testing problem, suppose that in a standard group testing setup where there are $n_0$ items, $t_0$ tests, $d_0$ defectives and the measurement matrix is Bernoulli$(p)$-generated, $\mathcal{A}$ guarantees a success probability of at least $\Theta(n_0,t_0,d_0,p)$. Then, when applying the algorithm $\mathcal{A}$ to our recovery problem, the success probability can be bounded as:
    \begin{align}\label{eq: general bound}
        \Pr(succ)\geq \displaystyle\prod_{i=1}^t\Theta(\sigma_i,s,\bar{\varphi}_i,\nu_i),
    \end{align}
    where $\nu_i = \binom{n-\varphi_i-1}{d-1}/\binom{n}{d}.$ The relation of $s, d, t, v_i$, and $\sigma_i$ can be expressed as
\begin{align}
        t-t\Upsilon(d)\left(1-\dfrac{s-1}{2d}\dfrac{a^2}{b-a^2}\right)\geq \displaystyle\sum_{i=1}^t(1-\nu_i)^{\sigma_i}\geq t-t\Upsilon(d), \label{eqn:generalSuccessBound}
    \end{align}
  
\end{prop}

We defer the proof of Proposition \ref{theo: general_bound} to Appendix \ref{proof for general bound}. Additionally, it should be noted further that Equation~\eqref{eqn:generalSuccessBound} can be simplified to
    \begin{equation}        
        \max_{i\in[t]}(1-\nu_i)^{\sigma_i}\geq 1-\Upsilon(d) \mbox{ and }  \min_{i\in[t]}(1-\nu_i)^{\sigma_i}\leq1-\Upsilon(d)\left(1-\dfrac{s-1}{2d}\dfrac{a^2}{b-a^2}\right).    
    \end{equation}

With Proposition \ref{theo: general_bound} established, we can now complete the proof of \cref{theo: universal bound for GT} as follows.

\begin{proof}[Proof of \cref{theo: universal bound for GT}]
    By applying Proposition \ref{theo: general_bound}, for any matrix $\mG$ with missing entries, we have the probability of recovering using algorithm $\mathcal{A}$ is bounded by 
    $$\displaystyle\prod_{i=1}^t\Theta(\sigma_i^\mG,s,\bar{\varphi}_i^M,\dfrac{\binom{n - \varphi_i^\mG - 1}{d- 1}}{\binom{n}{d}} )\geq \left[ \min_{(\sigma, \bar{\varphi}, \varphi) \in \mathcal{K}} \Theta\left( \sigma, s, \bar{\varphi}, \dfrac{\binom{n - \varphi - 1}{d - 1}}{\binom{n}{d}} \right) \right]^t  $$
    where $\sigma_i^\mG,\bar{\varphi}_i^\mG, \varphi_i^\mG$ are the number of missing cells, the number of missing cells with initial value of 1, and the number of cells with value 1 in row $i$ of $\mG$, respectively. Thus, the unsuccess probability of one sampling round is lower bounded by
    $$1 - \left[ \min_{(\sigma, \bar{\varphi}, \varphi) \in \mathcal{K}} \Theta\left( \sigma, s, \bar{\varphi}, \dfrac{\binom{n - \varphi - 1}{d - 1}}{\binom{n}{d}} \right) \right]^t \bar{\Theta}(n, t, d, p)$$
    Since we are considering $\hat{r}$ sampling round, the success probability of solving the group testing  problem via the recovered matrix is bounded by 
    $$1 - \left\{ 1 - \left[ \min_{(\sigma, \bar{\varphi}, \varphi) \in \mathcal{K}} \Theta\left( \sigma, s, \bar{\varphi}, \dfrac{\binom{n - \varphi - 1}{d - 1}}{\binom{n}{d}} \right) \right]^t \bar{\Theta}(n, t,d, p) \right\}^{\hat{r}}.$$ This completes the proof.
\end{proof}
\begin{remark}\label{remark: go to 1}
    If both $\mathcal{A}$ and $\mathcal{B}$ guarantee a positive success probability for all group testing setups, then \cref{theo: universal bound for GT} implies that as the number of sampling rounds increases, the lower bound on the success probability of the group testing problem using the recovered matrix will converge to 1.
\end{remark}
\begin{corollary}\label{last cor}
    For $\mathcal{A}$ and $\mathcal{B}$ as defined in \cref{theo: universal bound for GT}, if we want the success probability of solving the group testing problem via the recovered matrix to exceed $1-\varepsilon$ with $\vartheta^{\binom{n}{s}}\leq \varepsilon$, then the required number of sampling rounds is at least:
    \begin{align*}
        \hat{r}\geq \log_{\vartheta}(\varepsilon),
    \end{align*}
    where $\vartheta=1 - \left[ \min_{(\sigma, \bar{\varphi}, \varphi) \in \mathcal{K}} \Theta\left( \sigma, s, \bar{\varphi}, \dfrac{\binom{n - \varphi - 1}{d - 1}}{\binom{n}{d}} \right) \right]^t \bar{\Theta}(n, t, d, p)$.
    
\end{corollary}

\subsection{Instantiations}
\label{sub:instationations}

\par From here on, for any algorithm $\mathcal{A}$, we denote $\hat{\Pr}_{\mathcal{A}}(succ)$ as the success probability of solving a noiseless group testing with a Bernoulli($p_0$) test design, $n_0$ items, $t_0$ tests, $d_0$ defectives with the $\mathcal{A}$ algorithm. Additionally, we denote $\Pr_{\mathcal{A}}(succ)$ as the success probability of our matrix completion problem when solving with the $\mathcal{A}$ algorithm. In the next part, we will use our general bound for traditional algorithms for Group Testing. These algorithms include COMP (Combinatorial Orthogonal Matching Pursuit), which declares all items not seen in a negative test as defective; DD (Definite Defectives), which identifies items that must be defective because they are the only unexplained ones in a positive test and SSS (Smallest Satisfying Set), which finds the smallest set of items consistent with the test outcomes via linear programming (More details of these algorithms can be found in Appendix~\ref{subsec: Tra Al in GT}). Notice, however, that among the traditional algorithms there is also SCOMP (Sequential COMP), but due to its complex nature, no success probability bound has been established for it in the Bernoulli group testing setting. Therefore, we do not analyze them here.  
\subsubsection{The COMP algorithm}
\par From the work by \cite{Aldridge_2014},  for noiseless group testing with a Bernoulli($p$) test design, the success probability of the COMP algorithm is bounded by
\begin{align}\label{eq: GeneralBound COMP 1}
    \hat{\Pr}_{COMP}(succ)\geq 1-(n-d)\left(1-p(1-p)^d\right)^t
\end{align}
Thus, by letting $\Theta$ equal to the LHS of \cref{eq: GeneralBound COMP 1}, Proposition \ref{theo: general_bound} gives us the bound for our completion problem using the COMP algorithm as follow
\begin{align}\label{bound fo COMP Psucc}
    \Pr_{COMP}(succ)\geq \displaystyle\prod_{i=1}^t\left[1-(\sigma_i-\bar{\varphi}_i)\left(1-\nu_i(1-\nu_i)^{\bar{\varphi}_i}\right)^s\right]
\end{align}
It is important to note that the bound presented on the right-hand side of \cref{bound fo COMP Psucc} may suffer from numerical instability in practical settings, due to the multiplication of a large number of probabilities. In certain cases, this issue can be mitigated by deriving a more combinatorial bound. Specifically, when the number of item per test is constant (i.e. each row in $\mM$ having the same number of number 1), we can establish the following alternative bound.
    \begin{align}\label{special bound for COMP Psucc_}
        \Pr_{COMP}(succ)\geq 1-\bar{r}\left[\dfrac{\binom{n-1}{d}+\binom{n-1}{d-1}-\binom{d-\pi-1}{d-1}}{\binom{n}{d}}\right]^s=1-\bar{r}\left[1-\dfrac{\binom{n-\pi-1}{d-1}}{\binom{n}{d}}\right]^s
    \end{align}
    Where $\bar{r}$ is the number of $\Psi_{\alpha}\in\Psi$ such that $\psi_{\alpha}=0$. Additionally, from the definition of $\sigma_i$ and $\bar{\varphi}_i$, we also have: $$\bar{r}=\displaystyle\sum_{i=1}^t[\sigma_i-\bar{\varphi}_i]$$
It can be observed that the bound presented in \cref{special bound for COMP Psucc_} is more numerically stable than the one in \cref{bound fo COMP Psucc}. Nevertheless, when tested with a large number of samples, the two bounds yield similar results. The proof of \cref{special bound for COMP Psucc_} is provided in \cref{sec: bet bound for COMP}, and the corresponding experimental results are shown in \cref{fig:Prob_Bound_algo}.
\subsubsection{The DD algorithm}
Let us denote:
\begin{align*}
    &q_0(p,d)=(1-p)^d\\
    &q_1(p, d)=p(1-p)^{d-1}\\
    &b(d;n,t)=\binom{n}{d}t^d(1-t)^{n-d}\\
    &\Xi(n,t,m_0,p,d)=n (1 - p)^{m_0} \left( \exp \left( \frac{(t - m_0) p q_1(p,d)}{1 - q_0(p,d)} \right) - 1 \right) - \frac{q_1(p,d) (t - m_0)}{1 - q_0(p,d)}
\end{align*}
Following \cite{Aldridge_2014}, given a Bernoulli(p) test design, the probability of success under the DD algorithm
satisfies
\begin{align*}
    \hat{\Pr}_{DD}(succ)\geq \sum_{m_0=0}^{t} b(m_0; t, q_0(p,d)) \max \left[ 0, 1 - d \exp \left( \Xi(n,t, m_0,p,d) \right) \right]
    .
\end{align*}
Now, by letting $\Theta$ equal to the RHS of the bound, we can use Proposition \ref{theo: general_bound} to give a lower bound for our completion problem when solving with the DD problem as follow:
\begin{align*}
    \Pr_{DD}(succ)\geq \displaystyle\prod_{i=1}^t\left\{\sum_{m_0=0}^{s} b(m_0; s, q_0(\nu_i,\bar{\varphi}_i)) \max \left[ 0, 1 - \bar{\varphi}_i \exp \left( \Xi(\sigma_i,s, m_0,\nu_i,\bar{\varphi}_i) \right) \right]\right\}
\end{align*}

\subsubsection{The SSS algorithm}
\par For noiseless group testing with a Bernoulli(p) test design, the success probability of the SSS algorithm is given by \cite{Aldridge_2014} as
\begin{align*}
    \hat{\Pr}_{SSS}(succ)\geq 1 - d (1 - \mathcal{L}(p, d, d-1, d-1))^T - \sum_{B=0}^{d-1} \binom{d}{B} \binom{n-d}{d-B} (1 - \mathcal{L}(p, d, d, B))^T,
\end{align*}
where
\begin{align*}
    \mathcal{L}(p, d, L, B) = (1 - p)^d + (1 - p)^L - 2(1 - p)^{d+L-B}.
\end{align*}
Now, by letting $\Theta$ equal the RHS of the bound given by \cite{Aldridge_2014}, we yield the bound for our completion problem when solving with the SSS algorithm as follow: 
\begin{align*}
    \Pr_{SSS}(succ)\geq \displaystyle\prod_{i=1}^t\left[1 - \bar{\varphi}_i (1 - \mathcal{L}(\nu_i, \bar{\varphi}_i, \bar{\varphi}_i-1, \bar{\varphi}_i-1))^s - \sum_{B=0}^{\bar{\varphi}_i-1} \binom{\bar{\varphi}_i}{B} \binom{\sigma_i-\bar{\varphi}_i}{\bar{\varphi}_i-B} (1 - \mathcal{L}(\nu_i, \bar{\varphi}_i, \bar{\varphi}_i, B))^s\right]
\end{align*}

\section{Simulations}

\begin{figure}[!ht]
    \centering
    \includegraphics[width=0.8\linewidth]{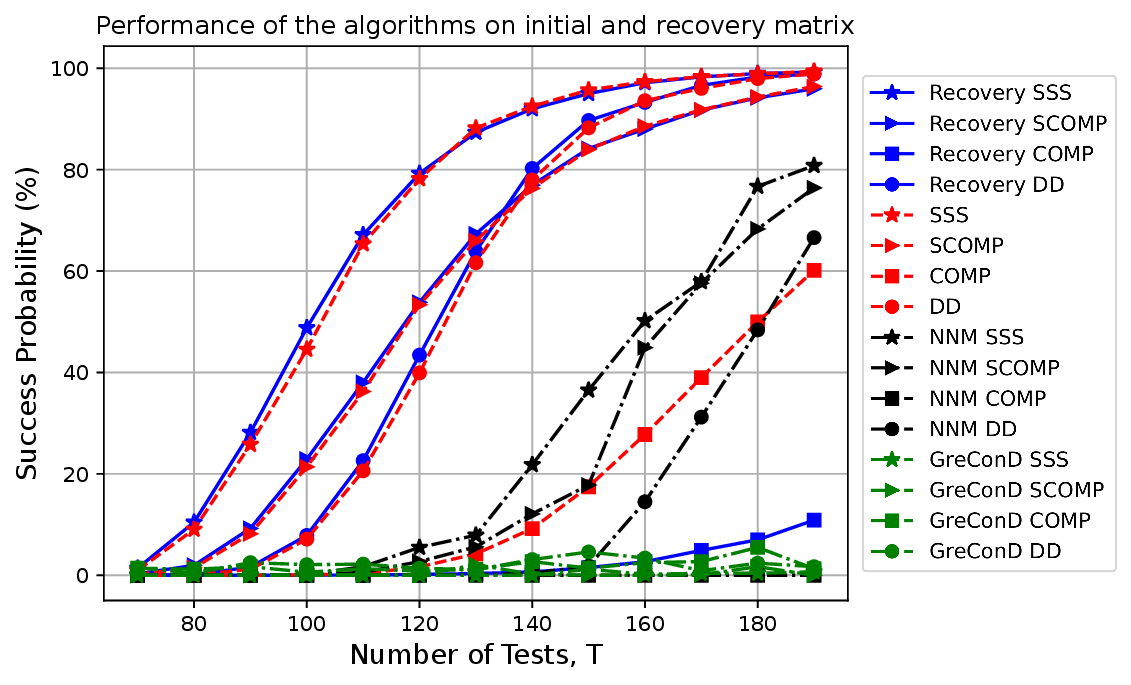}
    \caption{Performance of the COMP, SCOMP, DD, and SSS algorithms by using the recovery matrix as a measurement matrix after recovering the missing entries the missing matrix $\mG$ compared to the measurement matrix $\mM$ (the actual one). Here, we use the noiseless group testing set up with a Bernoulli test design with $n=500, d=10,p=0.1$. Additionally, each cell is missing with probability 0.1, and $s=10$ samples are used to recover the missing matrix. In the figure, Red: success rate on initial matrix, Black: success rate with measurement matrix recovered using NNM algorithms (In the legend, 'NNM + Algo A' refers to the missing matrix recovery process using NNM, followed by the application of Algorithm A for group testing), Green: success rate with measurement matrix recovered using GreConD (In the legend, 'GreConD + Algo A' refers to the missing matrix recovery process using GreConD, followed by the application of Algorithm A for group testing), Blue: success rate with measurement matrix recovered using our proposed method (In the legend, 'Recovery + Algo A' refers to the missing matrix recovery process using our proposed method, followed by the application of Algorithm A for group testing). The SSS, SCOMP, COMP, and DD algorithms are represented by stars,  triangles, squares, and circles, respectively.}
    \label{fig:Per_algo}
\end{figure}

\begin{figure}[!ht]
    \centering
    \includegraphics[width=0.85\linewidth]{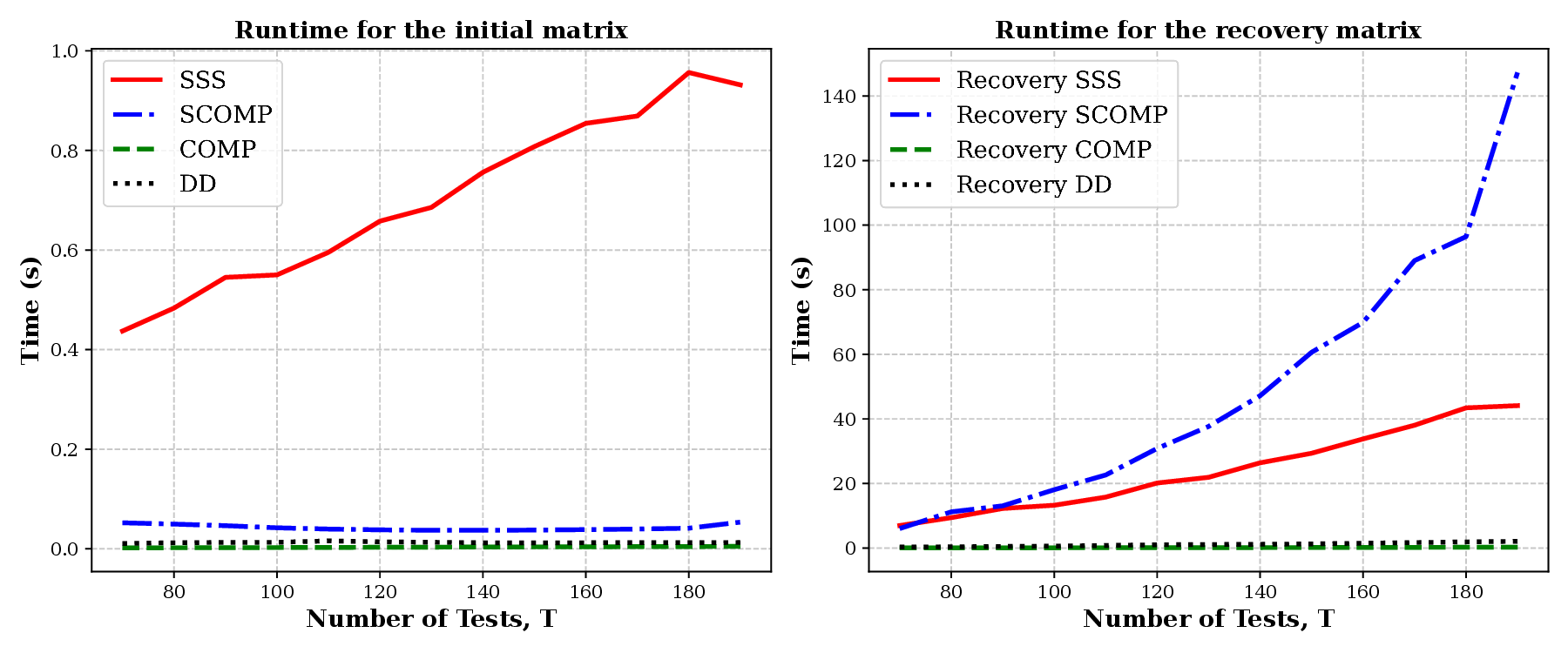}
    \caption{Computational time of the COMP, SCOMP, DD, and SSS algorithms by using the recovery matrix as a measurement matrix after recovering the missing entries the missing matrix $\mG$ compared to the measurement matrix $\mM$ (the actual one), when used in the same group testing problem. Here, we use the noiseless group testing setup with a Bernoulli test design with $n=500, d=10,p=0.1$. Additionally, each cell is missing with probability 0.1, and $s=10$ samples are used to recover the missing matrix. }
    \label{fig:time_algo}
\end{figure}

\begin{figure}[!ht]
    \centering
    \includegraphics[width=0.95\linewidth]{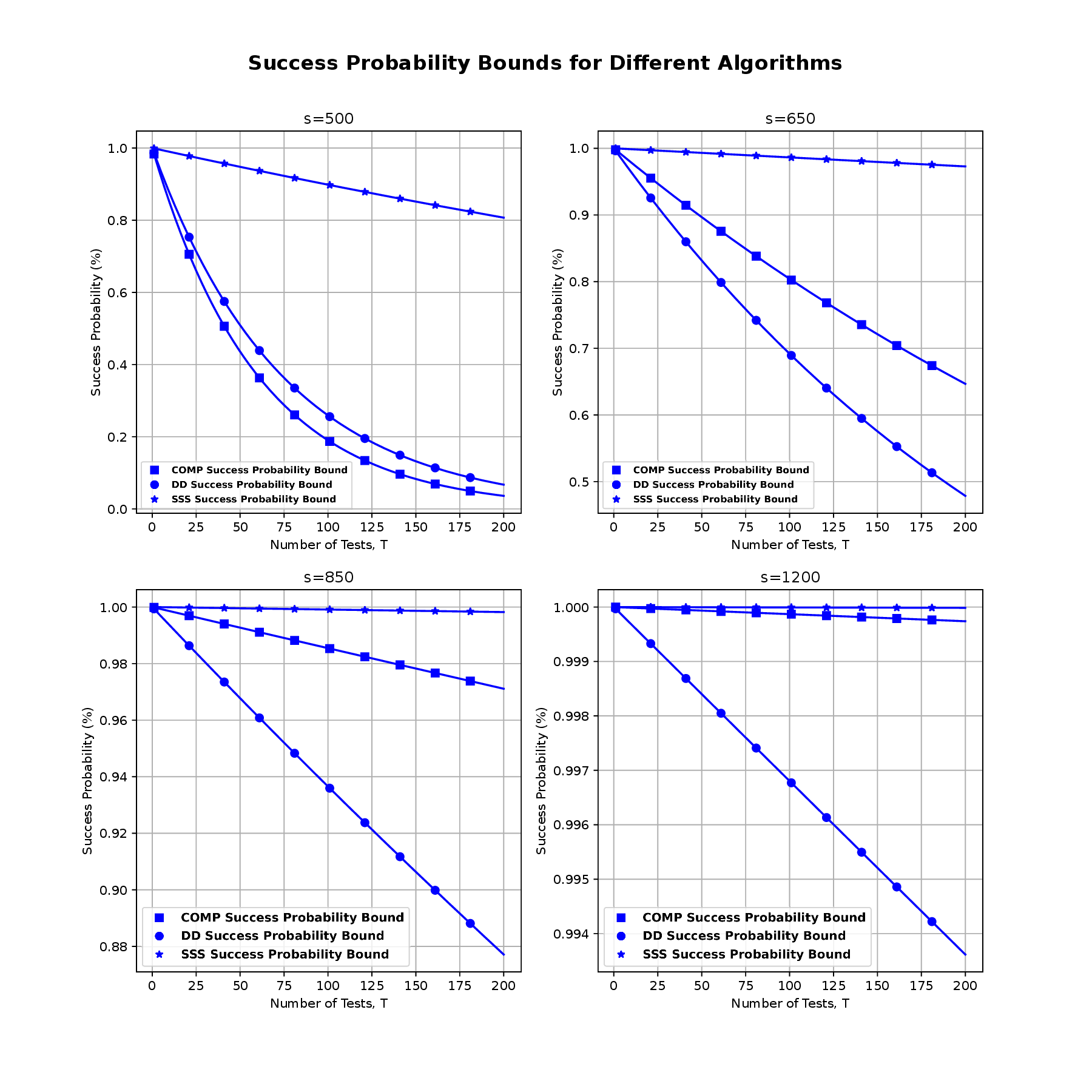}
    \caption{Visualizations of the success probability bounds for the SSS, COMP, and DD algorithms in a group testing framework with parameters \(n=300\), \(d=10\), \(q=0.05\), and \(p =0.1\) that are described in Section~\ref{sub:instationations}. The number of tests ranges from 1 to 200, and we evaluated the performance for four values of \(s\): 500, 650, 850, and 1200. The plotted results show the variation in success probability across different test counts, providing insights into the effectiveness of each algorithm under these conditions. In the figure, the SSS, COMP, and DD algorithms are represented by star, square, and circle markers, respectively. }
    \label{fig:Prob_Bound_algo}
\end{figure}

\par In this section, we conduct experiments to demonstrate that the performance of state-of-the-art algorithms (COMP, SCOMP, DD and SSS in~\cite{aldridge2014group}) for noiseless group testing remains consistent when applied to the recovery matrix (i.e., the matrix $\mG$ after recovering missing entries) compared to the measurement matrix $\mM$. Furthermore, as baselines to compare with our method, we evaluate the performance of the group testing problem when the measurement matrix is reconstructed using established Matrix Completion and Boolean Matrix Factorization techniques. In the Matrix Completion framework, we employ the Nuclear Norm Minimization (NNM) method. For the Boolean Matrix Factorization approach, we utilize the well-known GreConD algorithm, with the implementation adapted from~\cite{ignatov2021suboptimality}. Specifically, in our simulations we consider $n=500$ items, with $d=10$ defective ones. For each $t\in[70,190]$, we construct a binary test‐design matrix $\mathbf{M}\in\{0,1\}^{t\times n}$ by drawing each entry independently as $\mathrm{Bernoulli}\bigl(p\bigr)$, where we choose $p=1/d=0.1$. We then randomly select $s=10$ samples from $\mathcal{T}_d$ to generate a binary signal matrix $\mathbf{X}\in\{0,1\}^{n\times s}$. The noiseless measurements are given by $\mathbf{Y}=\mathbf{M}\odot\mathbf{X}$, where $\odot$ denotes Boolean (OR‐over‐rows) multiplication. To simulate missing entries, we delete each entry of $\mathbf{M}$ independently with probability $q=0.1$ to obtain the missing matrix $\mG$. We then apply, respectively, three matrix completion methods—NNM, GreConD, and our proposed algorithm—to $\mG$ to recover a reconstructed measurement matrix $\widehat{\mathbf{M}}$. For evaluation, we generate 1000 fresh test‐signal matrices $\{\mathbf{X}_{\mathrm{test}}^{(i)}\}_{i=1}^{1000}$, compute $\mathbf{Y}_{\mathrm{test}}^{(i)}=\widehat{\mathbf{M}} \odot \mathbf{X}_{\mathrm{test}}^{(i)}$ using the group‐testing operator $\odot$, and decode each $\mathbf{X}_{\mathrm{pred}}^{(i)}$ from $(\widehat{\mathbf{M}},\mathbf{Y}_{\mathrm{test}}^{(i)})$. The final accuracy is the fraction of trials for which $\mathbf{X}_{\mathrm{pred}}^{(i)}\equiv \mathbf{X}_{\mathrm{test}}^{(i)}$.

\par For each algorithm, we first report results in the ideal case where the measurement matrix is fully observed, i.e.\ we directly solve the group testing problem using the original measurement matrix $\mM$ (red curve in \cref{fig:Per_algo}). Next, we address the scenario where $\mM$ has missing entries: we recover $\mM$ using three methods—Nuclear Norm Minimization (black), GreConD (green), and our proposed algorithm (blue)—and plot their respective performances (black, green, and blue curves in \cref{fig:Per_algo}).

\par Overall, as demonstrated in \cref{fig:Per_algo}, our proposed method consistently outperforms the standard Nuclear Norm Optimization method and the GreConD method across all algorithms. This superior performance can be attributed to the fact that our method is able to incorporate the structural information of the group testing problem during the matrix recovery process. It is also worth noting that we experimented with using SVD as an alternative recovery method for the Matrix Completion scheme. However, the resulting matrices failed to solve any of the 1000 tests across all algorithms, resulting in a success probability of zero. For this reason, we do not include SVD results in the figure. Moreover, for the SSS, SCOMP and DD algorithms, our proposed method achieved success probabilities nearly identical to those obtained using the original matrix, further highlighting its effectiveness. Lastly, as shown in \cref{fig:Per_algo}, GreConD performs worse than the other methods, which can be partially explained by the fact that in the BMF scheme each algorithm must recover $\mG$ from scratch and cannot exploit the entries of $\mG$ that remain non-missing.

\par  Additionally, in \cref{fig:Prob_Bound_algo}, we plot the success-probability bounds for three state-of-the-art algorithms: COMP, DD, and SSS that are described in Section~\ref{sub:instationations}. These bounds are presented and proven in detail in \cref{sec: Gen bound for SP}. Specifically, we set $n=300$, $d=10$, $q=0.05$, and $p=0.1$, vary the number of tests $t$ from $1$ to $200$, and evaluate the results for four sample sizes $s=500,\,650,\,850,$ and $1200$. In \cref{sec: bet bound for COMP} we derive a numerically stable bound for COMP and add it to \cref{fig:Prob_Bound_algo_spec}, yielding the final figure with all traditional‐algorithm success bounds. Moreover, under the same experimental configuration as in \cref{fig:Per_algo}, \cref{fig:time_algo} presents the runtime of each algorithm when using our proposed method.

\section{Conclusion}
\label{sec:cls}

We consider a variant of the matrix completion problem in group testing. Instead of using the rank of the measurement matrix to recover it from the missing matrix, we utilize a number of observed input and outcome vectors. In particular, given the missing matrix $\mG$ and the number of input and outcome vectors observed, we construct a converted missing matrix $\Gamma$ and a converted missing vector $\bV$ such that $\Gamma \odot \bPsi = \bV$ with no duplicated rows in $\Gamma$, where $\bPsi$ is the representation vector of all missing entries in $\mG$. More importantly, we have shown that the information gain from the converted missing matrix $\Gamma$ and converted missing vector $\bV$ is equivalent to that of the missing matrix $\mG$ and the set of observed samples. Therefore, to reconstruct the measurement matrix, one only needs to reconstruct $\bPsi$ from $\Gamma$ and $\bV = \Gamma \odot \bPsi$.

Since the more rows $\Gamma$ has, the better the chance we have of recovering $\bPsi$, we derive the exact and approximate expected number of rows $\Gamma$. Unfortunately, in some cases, it is impossible to recover the missing entries regardless of the number of input and outcome vectors observed. This behavior should be studied in the future.

\section*{Acknowledgement}
This research is funded by the University of Science, VNU-HCM, Vietnam, under grant number CNTT 2024-22 and uses the GPUs provided by the Intelligent Systems Lab at the Faculty of Information Technology, University of Science, VNU-HCM, Vietnam.

\appendices

\section{Proof of Theorem \ref{thm:Formula_E}}
\label{append: A}

 Before deriving the formula for the expectation of $\omega$ when $s \geq 1$, we present two additional lemmas.
\begin{lemma}\label{general_E_for}
Let $n, d, p, q,s$, $\mathbf{X}$, $\cT_d$ be defined in~\cref{sec:intro}; $\Upsilon$ be defined in \cref{Ups_equ} and $\Pr(\theta_=)$ be defined in \cref{def theta_=}. Then, we have:
    \begin{align*}
        \bbE[\omega]=\binom{s}{1}\Upsilon(d)-\dfrac{\binom{\binom{n}{d}-2}{s-2}\phi}{2\binom{\binom{n}{d}}{s}}+\displaystyle\sum_{c=3}^{s}(-1)^{c+1}\dfrac{\displaystyle\sum_{\operatorname{col}(\mathbf{X})\subseteq \cT_d}\displaystyle\sum_{\theta\subseteq\operatorname{col}(\mathbf{X}),|\theta|=c}\Pr(\theta_=)}{\binom{\binom{n}{d}}{s}}
    \end{align*}
\end{lemma}
\par where $\phi=\displaystyle\sum_{\bX_i, \bX_j\in \cT_d, \bX_i\neq \bX_j}\Pr(\bX_i=\bX_j)$.

\begin{proof}
\par Recall the definition of $\bbE[\mX_c$] for an integer $c$ given in \cref{def expect chi_z}, we have:
\begin{align} 
    \bbE[\omega]&=\dfrac{\displaystyle\sum_{\operatorname{col}(\mathbf{X})\subseteq \cT_d}\bbE[\mathbf{X}_1]}{\binom{\binom{n}{d}}{s}}-\dfrac{\displaystyle\sum_{\operatorname{col}(\mathbf{X})\subseteq \cT_d}\bbE[\mathbf{X}_2]}{\binom{\binom{n}{d}}{s}}+\displaystyle\sum_{c=3}^{s}(-1)^{c+1}\dfrac{\displaystyle\sum_{\operatorname{col}(\mathbf{X})\subseteq \cT_d}\bbE[\mathbf{X}_c]}{\binom{\binom{n}{d}}{s}} \label{eq:omega:1} \\
    &=\dfrac{\displaystyle\sum_{\operatorname{col}(\mathbf{X})\subseteq \cT_d}\displaystyle\sum_{\theta\subseteq\operatorname{col}(\mathbf{X}),|\theta|=1}\Pr(\theta_=)}{\binom{\binom{n}{d}}{s}}-\dfrac{\displaystyle\sum_{\operatorname{col}(\mathbf{X})\subseteq \cT_d}\displaystyle\sum_{\theta\subseteq\operatorname{col}(\mathbf{X}),|\theta|=2}\Pr(\theta_=)}{\binom{\binom{n}{d}}{s}}+\displaystyle\sum_{c=3}^{s}(-1)^{c+1}\dfrac{\displaystyle\sum_{\operatorname{col}(\mathbf{X})\subseteq \cT_d}\displaystyle\sum_{\theta\subseteq\operatorname{col}(\mathbf{X}),|\theta|=c}\Pr(\theta_=)}{\binom{\binom{n}{d}}{s}} \label{eq:omega:2}\\
    &=\displaystyle\sum_{\theta\subseteq\operatorname{col}(\mathbf{X}),|\theta|=1}\Pr(\theta_=)-\dfrac{\binom{\binom{n}{d}-2}{s-2}\displaystyle\sum_{\theta\subseteq \cT_d, |\theta|=2}\Pr(\theta_=)}{\binom{\binom{n}{d}}{s}}+\displaystyle\sum_{c=3}^{s}(-1)^{c+1}\dfrac{\displaystyle\sum_{\operatorname{col}(\mathbf{X})\subseteq \cT_d}\displaystyle\sum_{\theta\subseteq\operatorname{col}(\mathbf{X}),|\theta|=c}\Pr(\theta_=)}{\binom{\binom{n}{d}}{s}} \label{eq:omega:3}\\
    &=\binom{s}{1}\Upsilon(d)-\dfrac{\binom{\binom{n}{d}-2}{s-2}\phi}{2\binom{\binom{n}{d}}{s}}+\displaystyle\sum_{c=3}^{s}(-1)^{c+1}\dfrac{\displaystyle\sum_{\operatorname{col}(\mathbf{X})\subseteq \cT_d}\displaystyle\sum_{\theta\subseteq\operatorname{col}(\mathbf{X}),|\theta|=c}\Pr(\theta_=)}{\binom{\binom{n}{d}}{s}} \label{eq:omega:4}
\end{align}
    \par ~\cref{eq:omega:1} is obtained due to the inclusion-exclusion principle. By the linearity of expectation, we have $\bbE[\mathbf{X}_k]=\displaystyle\sum_{\theta\subseteq\operatorname{col}(\mathbf{X}),|\theta|=c}\bbE(\theta_=)$. But since $\bbE[\theta_=]$ can only take the value 1 or 0, it is equal to $\Pr(\theta_=)$. Substituting this back into the expected value and we will get~\cref{eq:omega:2}.
    \par \cref{eq:omega:3} is obtained by combining the two following facts:
    \begin{itemize}
        \item For any $\operatorname{col}(\mathbf{X}) \subseteq \cT_d$ and $\theta \subseteq \operatorname{col}(\mathbf{X})$ with $\theta=\{\mathbf{g}\}$, we have $\Pr(\theta_=)$ is the probability that $(\mathbf{g},\mathbf{\tau})$ being informative and thus this probability is equal to $\Upsilon(d)$. This leads to $\displaystyle\sum_{\theta\subseteq\operatorname{col}(\mathbf{X}),|\theta|=1}\Pr(\theta_=)$ is the same for all $\mathbf{X}\in \cT_d$ and is equal to $\binom{s}{1}\Upsilon(d)$.
        \item Let us define $G=\displaystyle\sum_{\operatorname{col}(\mathbf{X})\subseteq \cT_d}\displaystyle\sum_{\theta\subseteq\operatorname{col}(\mathbf{X}),|\theta|=2}\Pr(\theta_=)$. To calculate this sum, each time we select a set $\mX$ such that $\operatorname{col}(\mathbf{X})$ from $\cT_d$, we add $\Pr(\bX_i=\bX_j)$ to $G$ for every unordered pair $(\bX_i,\bX_j)$ in $\operatorname{col}(\mathbf{X})\times \operatorname{col}(\mathbf{X})$. However, we also can do the equivalent process as follows: For every unordered pair $(\bX_i, \bX_j)\in \cT_d\times \cT_d$, we count the number of ways to choose $\mathbf{X}$ such that $\bX_i,\bX_j\in \operatorname{col}(\mathbf{X})$ (denote this as the weight of $(\bX_i,\bX_j)$). We then add $\Pr(\bX_i=\bX_j)$ multiplied by its weight to $G$. After performing this operation for all unordered pairs $(\bX_i,\bX_j)\in \cT_d\times \cT_d$, we will obtain the same $G$ as defined above. Furthermore, when using with this equivalent process, since $\mathbf{X}$ is taken uniformly from $\cT_d$, we can conclude that all the weights are equal to $\binom{\binom{n}{d}-2}{s-2}$.
    \end{itemize}    
\par Since $\Pr(\bX_i=\bX_j)=\Pr(\bX_j=\bX_i)$ for all $\bX_i, \bX_j \in \cT_d$, we have $\displaystyle\sum_{\bX_i, \bX_j\in \cT_d, \bX_i\neq \bX_j}\Pr(\bX_i=\bX_j)=2\times\displaystyle\sum_{\theta\subseteq \cT_d, |\theta|=2}\Pr(\bX_i=\bX_j)$. Now by letting $\phi=\displaystyle\sum_{\bX_i, \bX_j\in \cT_d, \bX_i\neq \bX_j}\Pr(\bX_i=\bX_j)$,~\cref{eq:omega:4} is obtained. Additionally, $\phi$ could be considered as the expected amount of ordered pair $(\bX_i,\bX_j)\in \cT_d\times\cT_d$ such that $\bX_i \neq \bX_j$, $(\bX_i,\tau),(\bX_j,\tau)$ are both informative and are identical.
\end{proof}

\par Our next target is to calculate $\phi$. But before we do that we will go to the definition of the $\simm$ function. For two same dimensional $n \times 1$ vectors $\bU$ and $\tilde{\bV}$, let $\simm(\bU, \bV) := \bU^T \bV$ be the number of positions that two vectors agree. To prove this theorem, we first calculate the probability of two informative pair being identical.
\begin{lemma}
Let $t$ and $\mathbf{X}$ be defined in~\cref{sec:intro}. For some $\bX_i, \bX_j \in \operatorname{col}(\mathbf{X})$ and $\tau \in [t]$ such that $(\bX_i, \tau)$ and $(\bX_j, \tau)$ are informative, we have
\begin{equation}    
    \Pr(\bX_i=\bX_j)=\Upsilon(\simm(\bX_i,\bX_j)).
\end{equation}
\label{same_pair_lemma}
\end{lemma}

\begin{proof}
Denote $\Delta_1=\{\delta|(\bX_i)_{\delta}=(\bX_j)_{\delta}=1\}$, $\Delta_2=\{\delta|(\bX_i)_{\delta}=0, (\bX_j)_{\delta}=1 \text{ or } \delta|(\bX_i)_{\delta}=1, (\bX_j)_{\delta}=0 \}$ and $\Delta_3=\{\delta|(\bX_i)_{\delta}=(\bX_j)_{\delta}=0\}$. Hence, $|\Delta_1|=\simm(\bX_i, \bX_j)$. Furthermore, since the number of ones in $\bX_i$ and $\bX_j$ are $d$, we also have $|\Delta_2|=2d-2\simm(\bX_i,\bX_j)$ and $|\Delta_3|=n-2d+\simm(\bX_i,\bX_j)$. Now, for all $\delta_0 \in [n]$ the followings must hold:
\begin{itemize}
    \item If $\delta_0 \in [n]\cap \Delta_1$, then $g_{\tau \delta_0} = \blacksquare$ or $g_{\tau \delta_0}=0$. Additionally, there must exist $\delta_1 \in [n] \cap \Delta$ such that $g_{\tau \delta_1} = \blacksquare$.
    \item If $\delta_0 \in [n] \cap \Delta_2$ then $g_{\tau \delta_0}= 0$.
    \item If $\delta_0 \in [n]\cap \Delta_3$, then $\bX_i=\bX_j$ is independent of the value of $g_{\tau \delta_0}$. 
\end{itemize}
\par The first bullet point arises from the fact that both $\bX_i$ and $\bX_j$ are informative, while the second bullet point results from the condition $\bX_i=\bX_j$. The third bullet point is the consequence of the fact that both the condition of informative and $\bX_i=\bX_j$ are independent of the zero-cells of $\bX_i$ and $\bX_j$.
\par Thus, we have:
\begin{align*}
    \Pr(\bX_i=\bX_j)&=\left[\displaystyle\prod_{\delta_0 \in \Delta_1}\Pr(\tau_{\delta_0}=-1 \text{ or }\tau_{\delta_0}=0)-\displaystyle\prod_{\delta_0 \in \Delta_1}\Pr(\tau_{\delta_0}=0)\right]\times \left[\displaystyle\prod_{\delta_0 \in \Delta_2}\Pr(\tau_{\delta_0}=0)\right]\\
    &=\left\{(1-p+pq)^{\simm(\bX_i,\bX_j)}-\left[(1-p)(1-q)\right]^{\simm(\bX_i,\bX_j)}\right\}\times \left[(1-p)(1-q)\right]^{2d-2\simm(\bX_i,\bX_j)}\\
    &=\Upsilon(\simm(\bX_i,\bX_j))
\end{align*}
\par This completes the proof.
\end{proof}
\par By using \cref{same_pair_lemma}, we derive a formal formula for $\phi$ as follows:
\begin{lemma}\label{phi_cal}
Let $n, d $ be defined in~\cref{sec:intro}, $\Upsilon$ be defined in \cref{Ups_equ} and $\phi$ be defined in \cref{general_E_for}. Then, we have:
    \begin{align*}
        \phi=\binom{n}{d}\displaystyle\sum_{i=0}^{d-1}\binom{d}{i}\binom{n-d}{d-i}\Upsilon(i).
    \end{align*}
\end{lemma}
\begin{proof}
    \par Consider any vector $\beta_0 \in \cT_d$. The number of vectors $\beta_1 \in \cT_d$ such that $\simm(\beta_0,\beta_1)=i$ is $\binom{d}{i}\binom{n-d}{d-i}$ for all $i \in \{0,\dots,d-1\}$. Hence, the total number of pair $(\beta_1,\beta_2)\in \cT_d\times \cT_d$ such that $\simm(\beta_1, \beta_2)=i$ is exactly $\binom{n}{d}\binom{d}{i}\binom{n-d}{d-i}$ for all $i \in \{0,\dots,d-1\}$. By summing up all these quantities, we acquire $\phi$ as mentioned.
\end{proof}

\emph{Proof of~\cref{thm:Formula_E}:} By substituting~\cref{phi_cal} into~\cref{general_E_for}, Theorem~\ref{thm:Formula_E} is attained.
\\\\
Note that, by Theorem~\ref{thm:s=1}, $\bbE[\omega] = \Upsilon(d)$, and if we randomly select an element $\mathbf{g}$ from $\cT_d$, the probability of $(\mathbf{g}, \tau)$ being informative is also $\Upsilon(d)$. Even though the quantity $\Upsilon(d)$ is crucial to our bound in \cref{Expectation_bound}, its formulation does not provide clear intuition about how the function scales with changes in the missing probability $q$. If we define $\Upsilon(d) = b^d - a^d$, where $b = 1 - p + pq$ and $a = (1 - p)(1 - q)$, we can derive upper and lower bounds for $\Upsilon(d)$ that offer a more interpretable understanding of how it behaves as $q$ varies in the following lemma.
\begin{lemma}\label{bound for Ups}
    With probabilities $p,q$, we have
    \begin{align*}
        qd(1-p-q+pq)^{d-1} < \Upsilon(d) < qd(1-p+pq)^{d-1}
    \end{align*}
\end{lemma}
\begin{proof}
    We have
    \begin{align*}
        b^d-a^d=(b-a)(b^{d-1}+b^{d-2}a+\dots+ba^{d-2}+a^{d-1})
    \end{align*}
    But since $a<b$, we have 
    $$(b-a)da^{d-1}<b^d-a^d<(b-a)db^{d-1} $$
    Replacing $a ,b$ with $(1-p)(1-q)$ and $1-p+pq$, we complete the proof.
\end{proof}
To further demonstrate the bound given in \cref{bound for Ups}, we have drawn out the value of $\Upsilon(d)$ and its two bounds with $d=10$, $p=0.1$, and $q$ varies from 0.01 to 0.1 in \cref{fig: Bound for Ups}.
\begin{figure}[!ht]
    \centering
\scalebox{0.6}{\includegraphics[width=0.8\linewidth]{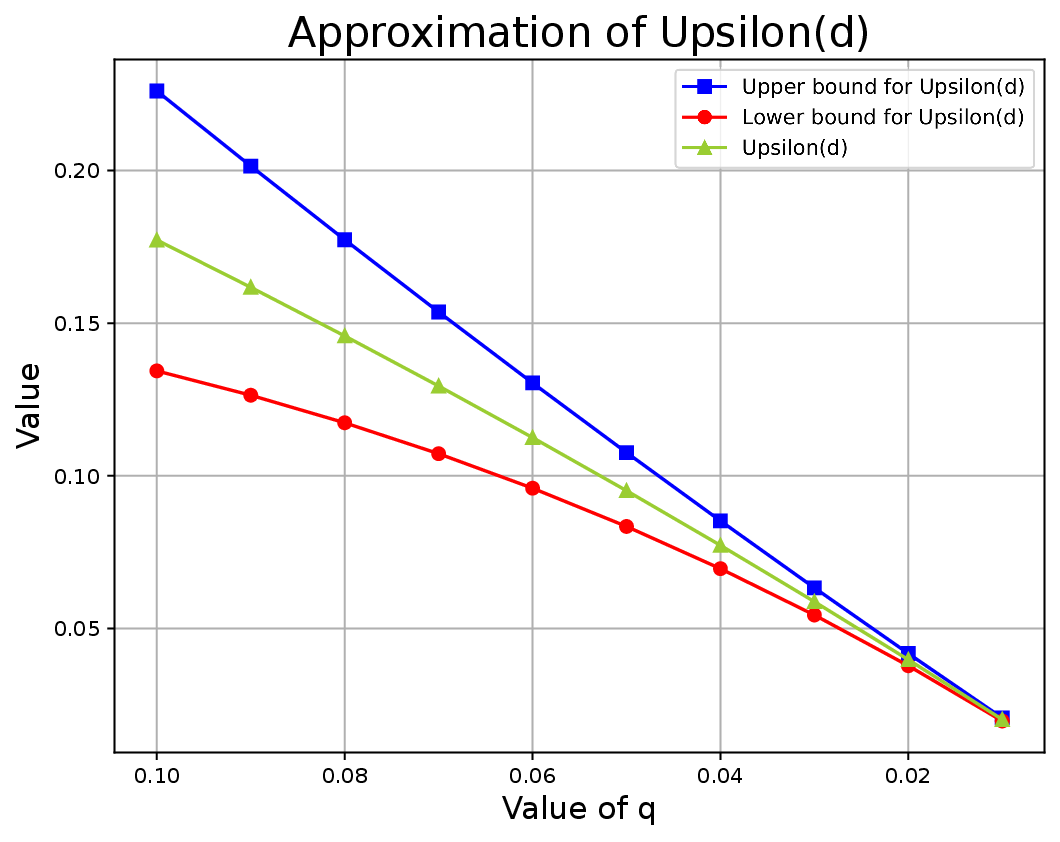}}
    \caption{The upper and lower bound for $\Upsilon(d)$ with $d=10$, $p=0.1$ and $q$ varies from $0.01$ to $0.1$.}
    \label{fig: Bound for Ups}
\end{figure}

\section{Proof of Theorem \ref{Expectation_bound}} \label{append: B}

\par First, to attain \cref{theo6 eq1}, we will go through the following lemma which would help to reduce the complexity of the formula introduced in \cref{theo5}.
\begin{lemma}\label{Inclusion-Exclusion Bound}
    Let $n, d, p, q,s$, $\mathbf{X}$, $\cT_d$ be defined in~\cref{sec:intro}; $\Upsilon$ be defined in \cref{Ups_equ} and $\Pr(\theta_=)$ be defined in \cref{def theta_=}. Then, we have:
    \begin{align}
    -\dfrac{\binom{n}{d}\binom{\binom{n}{d}-2}{s-2}}{2\binom{\binom{n}{d}}{s}}\cdot \left[\displaystyle\sum_{i=0}^{d-1}\binom{d}{i}\binom{n-d}{d-i}\Upsilon(i)\right]+\displaystyle\sum_{c=3}^{s}(-1)^{c+1}\dfrac{\displaystyle\sum_{\operatorname{col}(\mathbf{X})\subseteq \cT_d}\displaystyle\sum_{\theta\subseteq\operatorname{col}(\mathbf{X}),|\theta|=c}\Pr(\theta_=)}{\binom{\binom{n}{d}}{s}} &\leq 0, \label{lem:Inclusion-Exclusion Bound:1} \\
    \displaystyle\sum_{c=3}^{s}(-1)^{c+1}\dfrac{\displaystyle\sum_{\operatorname{col}(\mathbf{X})\subseteq \cT_d} \displaystyle\sum_{\theta\subseteq\operatorname{col}(\mathbf{X}),|\theta|=c}\Pr(\theta_=)}{\binom{\binom{n}{d}}{s}} &\geq 0. \label{lem:Inclusion-Exclusion Bound:2}
    \end{align}
\end{lemma}

\begin{proof}
\par First, let us denote:
\begin{align}
    &L[1]=s\Upsilon(d), \label{eqn:lem5:1} \\
    &L[2]=\dfrac{\binom{n}{d}\binom{\binom{n}{d}-2}{s-2}}{2\binom{\binom{n}{d}}{s}}\cdot \left[\displaystyle\sum_{i=0}^{d-1}\binom{d}{i}\binom{n-d}{d-i}\Upsilon(i)\right], \label{eqn:lem5:2} \\
    &L[3]=\displaystyle\sum_{c=3}^{s}(-1)^{c+1}\dfrac{\displaystyle\sum_{\operatorname{col}(\mathbf{X})\subseteq \cT_d} \displaystyle\sum_{\theta\subseteq\operatorname{col}(\mathbf{X}),|\theta|=c}\Pr(\theta_=)}{\binom{\binom{n}{d}}{s}}. \label{eqn:lem5:3}
    \end{align}
\par Hence, $\bbE[\omega]=L[1]-L[2]+L[3]$. The proof of~\cref{lem:Inclusion-Exclusion Bound:2} is as follows. For every uniformly random generated $\operatorname{col}(\mathbf{X})\subseteq \cT_d$, denote $\mathbf{X}=\left[\bX_1^{\top},\dots,\bX_s^{\top}\right]$. Let $C_i$ be a random variable that take the value 1 if and only if $(\bX_i,\mathbf{\tau})$ is informative and 0 otherwise. Additionally, for all $0<u<v\leq s$, we denote $C_{u,v}$ be a random variable that is 1 if and only if $\bX_u$ and $\bX_v$ are identical. Then we denote $Y= \displaystyle\sum_{i}C_i-\displaystyle\sum_{u<v}C_{u,v}.$
\par It is straightforward to see that $\bbE[Y]=L[1]-L[2]$. So all we need to do now is proving $\bbE[Y]<\bbE[\omega]$. To do this we will show that with every way of choosing $\mathbf{X}$ and the entries of row $\mathbf{\tau}$, we have $Y \leq \omega$. For a chosen $\mathbf{X}$ and row $\mathbf{\tau}$, we can partition $\mathbf{X}=\displaystyle\bigcup_{i=1,\dots,r+1} Z_i $ where for all $h \in \{1,\dots,r\}$ and for every $\mathbf{u}\neq \mathbf{v}\in Z_h$, we have $(\mathbf{v},\mathbf{\tau}),(\mathbf{u},\mathbf{\tau})$ are identical. Additionally, for all $i\neq j \in \{1,\dots,r\}$ and every $\mathbf{u}\in Z_i, \mathbf{v} \in Z_j$ then $(\mathbf{v},\mathbf{\tau}),(\mathbf{u},\mathbf{\tau})$ are not identical. Furthermore, for all $\mathbf{w} \in Z_{r+1}$ we have $(\mathbf{w},\mathbf{\tau})$ is not informative. Now, because of the definition of $\omega$ and $Y$, we have $\omega=r$, and $Y=\displaystyle\sum_{i=1}^{r}|Z_i|-\displaystyle\sum_{i=1}^{r}\binom{|Z_i|}{2}.$ Hence $Y \leq \omega$.
\par Using the same technique, we can also prove~\cref{lem:Inclusion-Exclusion Bound:1}.
\end{proof}

\par Now by applying \cref{Inclusion-Exclusion Bound}, along with some rearrangement of the variables, one can quickly derive \cref{theo6 eq1}. Nevertheless, the terms of the lower bound of $\bbE[h]/t$ in \cref{theo6 eq1} are somehow complex. To make it a simple one, we consider the case $n > (d+1)^2$. In addition to shorten the writing, we have the following notations:
\begin{equation} \label{omega_denote}
a=(1-p)(1-q),\quad b=1-p+pq,\quad
    \Omega(d,n)= \dfrac{\displaystyle\sum_{i=0}^{d-1}\binom{d}{i}\binom{n-d}{d-i}\Upsilon(i)}{\displaystyle\sum_{i=0}^{d-1}\binom{d}{i}\binom{n-d}{d-i}\Upsilon(d)}.
\end{equation}
\par Next we derive two crucial monotonic properties of $\Omega(d,n)$, which are shown on \cref{binom inq lemma} and \cref{Ups inc lemma}. 
\begin{lemma} \label{binom inq lemma}
    Let $n,d$ be positive integers such that $n>(d+1)^2$ then for all $i\in \{0,\dots,d-1\}$, we have:
    \begin{align}\label{binom inequality}
        \binom{d}{i}\binom{n-d}{d-i}>\binom{d}{i+1}\binom{n-d}{d-i-1}.
    \end{align}
\end{lemma}
\begin{proof}
    \par \cref{binom inequality} can be rewritten as:
    \begin{align*}
        \dfrac{d!}{i!(d-i)!}\dfrac{(n-d)!}{(d-i)!(n-2d+i)!}>\dfrac{d!}{(i+1)!(d-i-1)!}\dfrac{(n-d)!}{(d-i-1)!(n-2d+i+1)!}.
    \end{align*}
    \par Thus by simplifying, it is equivalent to:
    \begin{align*}
        (n-2d)(i+1)+2di+2i+1>d^2.
    \end{align*}
    \par This is true for all $n>(d+1)^2$, hence the proof completes.
\end{proof}
\begin{lemma} \label{Ups inc lemma}
    Let $\Upsilon$ be defined in \cref{Ups_equ}. Then , for all $i\in\{0,\dots,d-1\}$, we have:
    \begin{align*}
        \Upsilon(i)<\Upsilon(i+1).
    \end{align*}
\end{lemma}
\begin{proof}
    \par We have $\Upsilon(u)=a^{2d-2u}b^u-a^{2d-u}$. Thus, the derivative in terms of $u$ can be calculated as:
    \begin{align*}
        \frac{d}{du}\Upsilon(u)&=a^{2d-2u}b^u\ln(b)-2a^{2d-2u}b^u\ln(a)+a^{2d-u}\ln(a)\\
        &=a^{2d-2u}b^u[\ln(b)-\ln(a)]-\ln(a)a^{2d-2u}(b^u-a^u).
    \end{align*}
    \par But since we have $1>b>a>0$, we have $a^{2d-2u}b^u[\ln(b)-\ln(a)]>0$ and $-\ln(a)a^{2d-2u}(b^u-a^u)>0$. Thus, for all $0\leq u\leq d$ we have $\frac{d}{du}\Upsilon(u)>0$. This directly yields the desired property.
\end{proof}

\par Lastly, by taking advantage of the two monotone sequences mentioned by~\cref{binom inq lemma} and \cref{Ups inc lemma}, in addition with Chebyshev's sum inequality, we derive \cref{lem Omega}.
\begin{lemma} \label{lem Omega}
    Let $\Omega,a$ and $b$ be defined in \cref{omega_denote}. If $d,n$ are defined in \cref{sec:intro} and satisfy $n > (d+1)^2$, then we have:
    \begin{align*}
        \Omega(d, n)<\dfrac{a^2}{b-a^2}d^{-1}.
    \end{align*}
\end{lemma}

\begin{proof}
    \par Consider two real number sequences $a_0,\dots,a_{d-1}$ and $b_0,\dots,b_{d-1}$ such that $a_i=\Upsilon(i)$ and $b_i=\binom{d}{i}\binom{n-d}{d-i}$. Now by applying \cref{binom inq lemma} and \cref{Ups inc lemma}, we have: $$a_0<\dots<a_{d-1},$$ $$b_0>\dots>b_{d-1}.$$ Thus by using Chebyshev's sum inequality, we get
    \begin{align*}
        \Omega(d, n) \leq d^{-1}\cdot \dfrac{\left[\displaystyle\sum_{i=0}^{d-1}\binom{d}{i}\binom{n-d}{d-i}\right]\left[\displaystyle\sum_{i=0}^{d-1}\Upsilon(i)\right]}{\displaystyle\sum_{i=0}^{d-1}\binom{d}{i}\binom{n-d}{d-i}\Upsilon(d)} =d^{-1}\cdot \dfrac{\displaystyle\sum_{i=0}^{d-1}\Upsilon(i)}{\Upsilon(d)}.
    \end{align*}
    \par By substituting $\Upsilon(i)=a^{2d-2i}b^i-a^{2d-i}$ and simplifying the terms, we get:
    \begin{align*}
        \Omega(d,n)<d^{-1}\cdot\dfrac{a^2-a^3+(b-a)a\left(\frac{a^2}{b}\right)^d+\left(\frac{a}{b}\right)^d(a^3-ab)}{\left[1-\left(\frac{a}{b}\right)^d\right](b-a^2)(1-a)}.
    \end{align*}
    \par Furthermore, since $0<a^2<a<b<1$, we get $\left(\frac{a^2}{b}\right)^d<\left(\frac{a}{b}\right)^d$. Thus $$a^2-a^3+(b-a)a\left(\frac{a^2}{b}\right)^d+\left(\frac{a}{b}\right)^d(a^3-ab)<\left[1-\left(\frac{a}{b}\right)^d\right](a^2-a^3).$$ 
    \par Combining this with our bound for $\Omega(d,N)$ we get:
    \begin{align*}
        \Omega(d,n)<d^{-1}\cdot\dfrac{\left[1-\left(\frac{a}{b}\right)^d\right](a^2-a^3)}{\left[1-\left(\frac{a}{b}\right)^d\right](b-a^2)(1-a)}=d^{-1}\cdot\dfrac{a^2}{b-a^2}.
    \end{align*}
    \par This completes the proof.
\end{proof}

\emph{Proof of~\cref{Expectation_bound}.} By substituting \cref{lem Omega} back into \cref{theo6 eq1}, we immediately acquire \cref{sample-bound}, which proves \cref{Expectation_bound}.

\section{Proof for Theorem \ref{theo: universal bound for GT}}

\subsection{Proof for Theorem \ref{independent of success}}
\label{proof for theo idsuc}

\begin{proof}
    We know that $\Pr(succ)$ is the probability that we can recover $\Psi$. But since, $\Psi$ is partitioned into $t$ sets $\Psi[1],\dots,\Psi[t]$, each is being recovered in the corresponding problems $\mathcal{Q}_1,\dots,\mathcal{Q}_t$, hence:
    \begin{align}\label{prob inde eq}
        \Pr(succ)=\Pr(\mathcal{Q}_1 \text{ is success},\dots,\mathcal{Q}_t \text{ is success})
    \end{align}
    We are left to prove that $\Pr_{\mathcal{Q}_1}(succ),\dots,\Pr_{\mathcal{Q}_t}(succ)$ are pairwise independent. This can be shown by noticing that, given $i \in [t]$, due to how we define $\tilde{\Gamma}$ and $\mathcal{Q}_i$ previously, for any column $j$ such that $\sigma_1+\dots+\sigma_{i-1}+1\leq j \leq \sigma_1 +\dots+\sigma_i$ and row $r$, we have:
    \begin{align*}
        \tilde{\Gamma}[r,j]=0 \quad\forall r\notin [s(i-1)+1,si] 
    \end{align*}
    This directly lead to the fact that for any $u\neq v\in [t]$, the $s$ rows in problem $\mathcal{Q}_u$ will give no information for the problem $\mathcal{Q}_v$ and vice versa, thus showing that their success probability are independent. Combining this with $\cref{prob inde eq}$, \cref{independent of success} holds.
\end{proof}

\subsection{Proof for Lemma \ref{ber gen}}
\label{proof for ber gen}

\begin{proof}
    Consider $\alpha\in [t]$, denote $\mathcal{Z}_0,\mathcal{Z}_1$ to be the set of all columns of $\mM$ that contain the number 0 and 1 at row $\alpha$, respectively. \\
    As we have mentioned, we have $|\mathcal{Z}_1|=\varphi_\alpha$, $|\mathcal{Z}_0|=n-\varphi_\alpha-\sigma_\alpha$ and $|\Psi[\alpha]|=\sigma_\alpha$. Now let us consider an arbitrary cell of $\tilde{\Gamma}$ with coordinate $(\gamma,\delta)$ and denote $\Psi_{\hat{\delta}}$ to be the corresponding sample at column $\delta$ of $\mathcal{R}_\alpha$. Following the construction of $\tilde{\Gamma}$ as in \cref{alter def gamma}, we have the value at $\mathcal{R}(\gamma,\delta)$ being equal to 1 is equivalent to the value of $x_{\gamma}[j_{\hat{\delta}}]$ being 1 and, for all $t\in \mathcal{Z}_1$, $x_\gamma[t]=0$ . Formally, this can be expressed in terms of probability as:
    \begin{align}\label{prob =1}
        \Pr(R(\gamma,\delta)=1)=\Pr(x_\gamma[j_{\hat{\delta}}]=1, x_\gamma[t]=0\quad\forall t\in \mathcal{Z}_1)
    \end{align}\
    \par Our goal would be to calculate the right-hand side of \cref{prob =1}. We already know that the number of ways to choose $x_\gamma$ is $\binom{n}{d}$. On the other hand, the number of ways to to choose $x_\gamma$ such that $x_\gamma[j_{\hat{\delta}}]=1, x_\gamma[t]=0\quad\forall t\in \mathcal{Z}_1$ is the same as the number of way to choose $d-1$ numbers in $(\mathcal{Z}_0\cup\Psi[\alpha])\backslash \{\hat{\delta}\}$. This amount is $\binom{n-\varphi_\alpha-1}{d-1}$, combining this with with \cref{prob =1}, we have:
    \begin{align*}
        \Pr(R(\gamma,\delta)=1)=\dfrac{\binom{n-\varphi_\alpha-1}{d-1}}{\binom{n}{d}}
    \end{align*}
    This finishes the proof.
\end{proof}
\subsection{Proof for Proposition \ref{theo: general_bound}} \label{proof for general bound}
\begin{proof}
    \textbf{General bound. } Applying \cref{independent of success}, we have
    \begin{align*}
        \Pr(succ)=\displaystyle\prod_{i=1}^{t}\Pr_{\mathcal{Q}_i}(succ)
    \end{align*}
    Furthermore, by the definition of $\Theta$ and $\mathcal{A}$ combining with \cref{ber gen} , one can show 
    $$\Pr(\mathcal{Q}_i\text{ is success)}\geq \Theta(\sigma_i,s,\bar{\varphi}_i,\nu_i)\quad \forall i\in[t]$$
    These two facts yield the general bound in \cref{eq: general bound}.\\
    \textbf{Bounds for $\nu_i$ and $\sigma_i$.} Denote $h_i$ as a random variable that show the number of rows that have at least one test in the measurement matrix of $\mathcal{Q}_i$, for $i\in[t]$. Thus, we would have $$\displaystyle\sum_{i=1}^th_i=h$$. But since $\mathcal{Q}_i$ are proved to be independent of each other, $h_i$ are also independent of each other. Hence, we have
    \begin{align}\label{extra step}
        \sum_{i=1}^t\mathbb{E}[h_i]=\mathbb{E}[h]
    \end{align}
    Now, since the measurement matrix of $\mathcal{Q}_i$ has $\sigma_i$ columns and is $\nu_i$-Bernoulli generated, the probability of a row not being fully 0 is $1-(1-\nu_i)^{\sigma_i}$, hence we have $$\mathbb{E}[h_i]=s(1-(1-\nu_i)^{\sigma_i})\quad \forall i\in[t]$$
    Combining this with \cref{extra step} and \cref{Expectation_bound}, we complete the proof.
\end{proof}

\section{Illustration of Theorem~\ref{thm:construction}}
\label{app:alg}

\subsection{Algorithm}

The procedure in Theorem~\ref{thm:construction} can be parsed to Algorithm~\ref{alg:missing_matrix_vector}.

\begin{algorithm}[H]
\caption{Construction of Converted Missing Matrix $\Gamma$ and converted missing vector $\bV$}

\begin{algorithmic}[1]
\State \textbf{Input:} Matrix $\mM$, $\mathbf{X}$, number of tests $t$
\State \textbf{Output:} Matrix $\Gamma$, vector $\mathbf{v}$
\State Initialize an empty $0 \times r$ matrix $\Gamma$ and a $0 \times 1$ vector $\bV$
\For{each pair $(\mathbf{x}, c) \in \operatorname{col}(\mathbf{X}) \times [t]$ that is informative}
    \State Initialize a row vector $\mathbf{g}$ of length $r$ with all zeros
    \For{each $z \in [r]$}
        \If{$i_z=c$ and $\mathbf{x}_{j_z} = 1$}
            \State Set $g_z = 1$
        \Else
            \State Set $g_z = 0$
        \EndIf
    \EndFor
    \State Append row $\mathbf{g}$ to matrix $\Gamma$
    \State Append $y_c$ to vector $\bV$
\EndFor
\State Remove duplicate rows from $\Gamma$ and corresponding entries from $\bV$
\State \Return $\Gamma$ and $\bV$
\end{algorithmic}

\label{alg:missing_matrix_vector}
\end{algorithm}

\begin{algorithm}
\caption{Construction of $\tilde{\Gamma}$ and $\tilde{\bV}$ }
\label{Algo2}
\begin{algorithmic}[1]
\State \textbf{Input:} Matrix $\mM$, $\mathbf{X}$, number of tests $t$
\State \textbf{Output:} Matrix $\tilde{\Gamma}$, vector $\tilde{\bV}$
\State Initialize an empty $0 \times r$ matrix $\tilde{\Gamma}$ and a $0 \times 1$ vector $\tilde{\bV}$
\For{each row $c$ in $\mM$}
    \For{each $\mathbf{x}$ in $\operatorname{col}(\mathbf{X})$}
        \If{$(\mathbf{x},c)$ is informative}
            \State Initialize a row vector $\mathbf{g}$ of length $r$ with all zeros
            \For{each $z \in [r]$}
                \If{$i_z = c$ and $x_{j_z} = 1$}
                    \State Set $g_z = 1$
                \Else
                    \State Set $g_z = 0$
                \EndIf
            \EndFor
            \State Append row $\mathbf{g}$ to matrix $\tilde{\Gamma}$
            \State Append $y_c$ to vector $\tilde{\bV}$
        \Else
            \State Initialize a row vector $\mathbf{g}$ of length $r$ with all zeros
            \State Append row $\mathbf{g}$ to matrix $\tilde{\Gamma}$
            \State Append $0$ to vector $\tilde{\bV}$
        \EndIf
    \EndFor
\EndFor
\State \textbf{Return} $\tilde{\Gamma}$, $\tilde{\bV}$
\end{algorithmic}
\end{algorithm}

\subsection{Special case for success probability when using the COMP algorithm in a constant-number-of-items per test circumstance}
\label{sec: bet bound for COMP}

\par In this section, we will work with the case where the number of ones in each row of $\mM$ is a constant $\pi$. We will show that in this circumstance, a more numerically friendly bound for the success probability for the COMP algorithm can be yield comparing to \cref{bound fo COMP Psucc}.

\begin{definition}
    We denote a missing cell $\Psi_{\alpha}$ as $\bar{0}$ if $\psi_{\alpha}=0$ and  as $\bar{1}$ otherwise. Following this, our missing matrix $\mG$ is generated with each cell having the value of $0,1,\bar{0},\bar{1}$ with probabilities $(1-p)(1-q),p(1-q),(1-p)q,pq$ respectively.
\end{definition}

\begin{definition}[Coverable set $\mathcal{T}$]
    A cell in $\Psi_j\in\Psi$ is considered \textbf{coverable} if for any row $i$ of $\Gamma$, we have:
    \begin{itemize}
        \item $\Gamma[i,j]=0$ or
        \item $\exists t\in [r]$ such that $\Gamma[i,t]=1$ and $\psi_t=1$.
    \end{itemize}
    The set of all convertable cells is denoted as $\mathcal{T}$.
\end{definition}

\begin{definition}[True set $\mathcal{V}$]
    Consider a missing cell $\Psi_{\alpha}\in\Psi$, $\Psi_{\alpha}$ is called \textbf{true} if there exist $x\in\operatorname{col}(\mathbf{X})$ such that:
    \begin{itemize}
        \item $x_{j_{\alpha}}=1$
        \item $\forall t\in[n]$ such that $x_t=1$, $M[i_{\alpha},t]\in\{0,\bar{0}\}$
    \end{itemize}
    The set of all true cells is denoted as $\mathcal{V}$.
\end{definition}

Denote $\Psi^{0}=\{\Psi_{\alpha}\in\Psi|\psi_{\alpha}=0\}=\{\Psi^{0}_1,\dots,\Psi^{0}_{\bar{r}}\}$.

\begin{lemma}\label{COMP fail}
    The algorithm COMP will fail to reconstruct $\mM$ if there exist $i\in [r]$ such that $\Psi_i\in \mathcal{T}$ and $\psi_i=0$.
\end{lemma}

\begin{proof}
    Following the definition of $\mathcal{T}$ and the COMP algorithm. The algorithm will not be able to classify $\Psi_i$ to be definitely negative. Thus, this would lead to the algorithm outputting $\psi_i=1$, which will be incorrect. This leads to the COMP algorithm not be able to reconstruct the original matrix $\mM$.
\end{proof}
\begin{lemma}\label{equivalent T V}
    For any $\Psi_{\alpha}\in\Psi$, we have:
    \begin{align*}
        \Psi_{\alpha}^0\in\mathcal{T} \Longleftrightarrow\Psi^0_{\alpha}\notin\mathcal{V}. 
    \end{align*}
\end{lemma}

\begin{proof}
    This can be derived directly from the definition of an informative pair and the construction of $\Gamma$.
\end{proof}

\par Let us denote $A_i$ to be the event that $\Psi_i^0\notin\mathcal{V}$ for all $i\in[\bar{r}]$. Now by applying \cref{COMP fail} and \cref{equivalent T V}, we will have:
\begin{align}\label{P ineq}
    \Pr_{COMP}(succ)=1-\Pr(fail)=1-\Pr(\bigcup_{i\in[\bar{r}]}A_i)\geq 1-\displaystyle\sum_{i\in[\bar{r}]}\Pr(A_i)
\end{align}

\par Our last target would be to calculate the value of $\Pr(A_i)$. This is done through the following lemma.

\begin{lemma}\label{calc P(A)}
    We have:
    \begin{align*}
        \Pr(A_{\alpha})=\left[\dfrac{\binom{n-1}{d}+\binom{n-1}{d-1}-\binom{n-\pi-1}{d-1}}{\binom{n}{d}}\right]^s \quad \forall \alpha\in[\bar{r}]
    \end{align*}
\end{lemma}

\begin{proof}
    For a sample $\bX\in\operatorname{col}(\mathbf{X})$. since $\Psi_{\alpha}^0$ does not belong to $\mathcal{V}$, one of the following two cases must happen:\\
        \textbf{Case 1. }$\bX_{j_{\alpha}}=0$. Since $x$ contains exactly $d$ ones, the probability of this case happening is $\dfrac{\binom{n-1}{d}}{\binom{n}{d}}$.\\
    \textbf{Case 2. }$\bX_{j_{\alpha}}=1$ and there exist $t\in[n], t\neq j_{\alpha}$ such that:
    \begin{itemize}
        \item $\bX_t=1$
        \item $\mG_{i_{\alpha},t}\in\{1,\bar{1}\}$
    \end{itemize}Since $\bX$ contains exactly $k$ ones, the probability of this case happening is $\dfrac{\binom{n-1}{d-1}-\binom{n-\pi-1}{d-1}}{\binom{n}{d}}$.\\
    Additionally, since $|\operatorname{col}(\mathbf{X})|=s$, we have the final probability of $A_{\alpha}$ happening is:
    $$\Pr(A_{\alpha})=\left[\dfrac{\binom{n-1}{d}+\binom{n-1}{d-1}-\binom{n-\pi-1}{d-1}}{\binom{n}{d}}\right]^s$$
\end{proof}

Lastly, by combining \cref{P ineq} and \cref{calc P(A)}, we yield the bound for P(succ) of the COMP algorithm as follow:

\begin{theorem} 
    The success probability of the COMP algorithm can be bounded as:
    \begin{align}\label{special bound for COMP Psucc}
        \Pr_{COMP}(succ)\geq 1-\bar{r}\left[\dfrac{\binom{n-1}{d}+\binom{n-1}{d-1}-\binom{d-\pi-1}{d-1}}{\binom{n}{d}}\right]^s=1-\bar{r}\left[1-\dfrac{\binom{n-\pi-1}{d-1}}{\binom{n}{d}}\right]^s,
    \end{align}
    where $\bar{r}$ is the number of $\Psi_{\alpha}\in\Psi$ such that $\psi_{\alpha}=0$. Additionally, from the definition of $\sigma_i$ and $\bar{\varphi}_i$ from the previous section, we also have: $$\bar{r}=\displaystyle\sum_{i=1}^t[\sigma_i-\bar{\varphi}_i]$$
\end{theorem}

\par Furthermore, we can see that the bound given in \cref{special bound for COMP Psucc} is more numerically friendly, as in real-life settings, $t$ could be very large and would cause the product in \cref{bound fo COMP Psucc} to have numerical errors. We have plotted out the value for each success probability bound in Fig.~\ref{fig:Prob_Bound_algo_spec}.

\begin{figure}[!ht]
    \centering
    \includegraphics[width=0.9\linewidth]{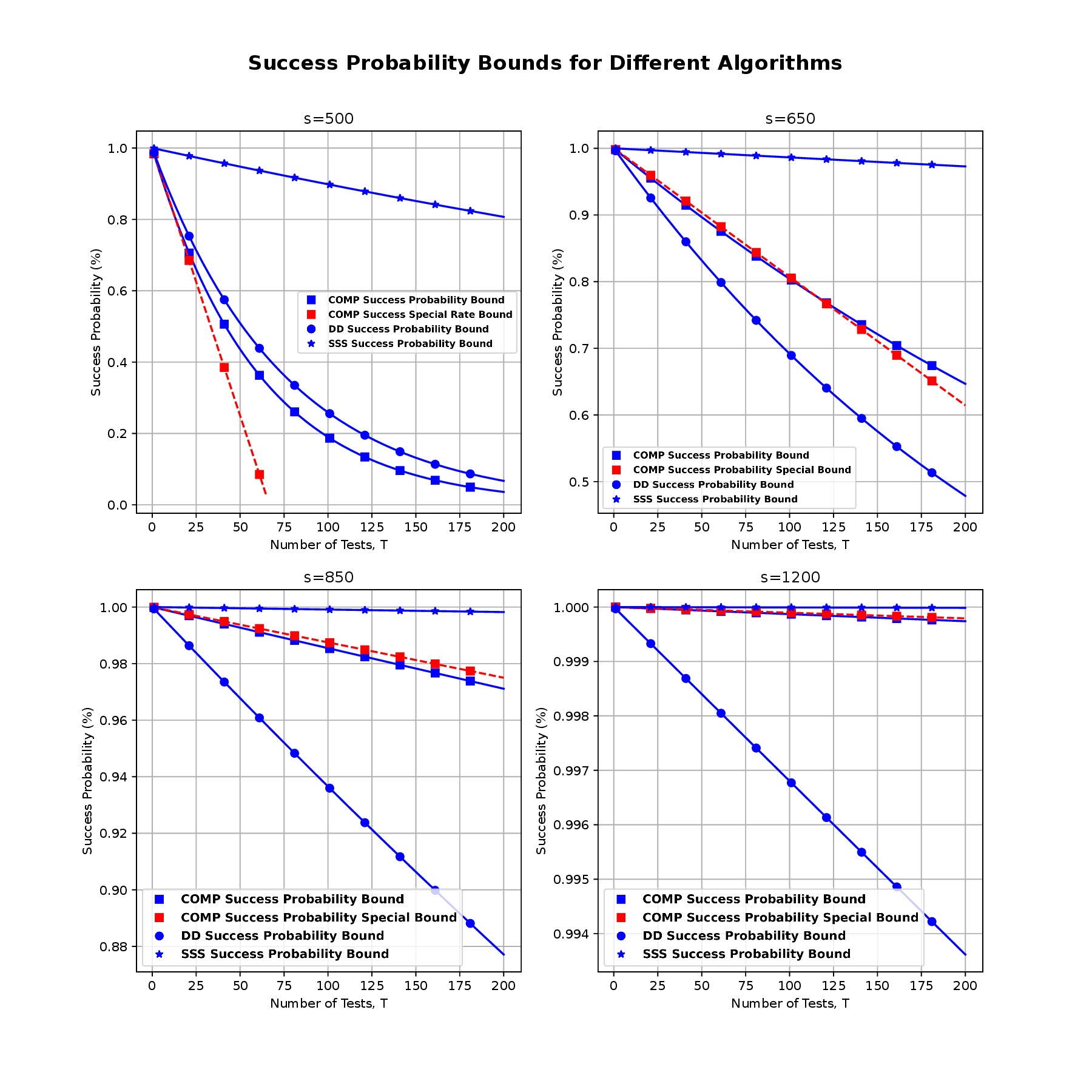}
    \caption{Visualization of the success probability bounds for the SSS, COMP, and DD algorithms in a group testing model with 300 items, 10 defectives, a missing entry rate of 0.05, and test inclusion probability 0.1. Varying the number of tests from 1 to 200 and evaluating four sample sizes (500, 650, 850, and 1200), we plotted how each algorithm’s chance of correct recovery scales with more tests. In the resulting figure, SSS is shown with star markers, COMP with squares, and DD with circles. Additionally, we also plot the special bound for the COMP algorithm introduced in Appendix \ref{sec: bet bound for COMP} in red.}
    \label{fig:Prob_Bound_algo_spec}
\end{figure}

\subsection{Feasible recovery}
\label{app:alg:feasible}

In this example, we show that it is feasible to recover all missing entries. Let us assume the measurement matrix $\mM$ and its missing matrix $\mG$ are as in~\cref{eqn:connectivity}. By Theorem~\ref{thm:construction}, the set $\psi=\{\psi_1,\psi_2,\psi_3,\psi_4,\psi_5\}$ is the solution to the group testing problem with the converted missing matrix being $\Gamma$ and the testing vector being $\bV$. Now by solving $\Gamma$ and $\bV$, we can recover the initial values of some of the missing entries of $\mM$: $\psi_{3}=0,\psi_{4}=0,\psi_{5}=1.$
\par However, we do not have enough information to recover $\psi_{1}, \psi_{2}$. To demonstrate that the larger the size of $\mathbf{X}$ and $\mathbf{Y}$, the greater the chance we will get at recovering the missing entries. Let us say another pair $(\bX_3,\bY_3)$ is added to $\operatorname{col}(\mathbf{X})\times\operatorname{col}( \mathbf{Y})$ where:
$$\begin{Bmatrix}
\bX_3=\left[0,1,1,0,0,0,0,0,0,0,0,0\right]^T & ; & \bY_3=[0,0,1,0,1,1,1,1,0]^T.
\end{Bmatrix}$$

Then we get $(\bX_3, 2)$ is informative. Hence, our converted missing matrix and converted missing vector will be modified as follows:
\begin{equation}
    \Gamma = \left[ \begin{array}{ccccccc}
    1 & 1 & 0 & 0 & 0 \\
    0 & 0 & 1 & 0 & 0 \\
    0 & 0 & 0 & 0 & 1 \\
    0 & 0 & 1 & 1 & 0 \\
    1 & 0 & 0 & 0 & 0 
    \end{array} \right], 
    \bV = \left[ \begin{array}{c}
    0\\
    0\\
    1\\
    0\\
    0
    \end{array} \right] 
\end{equation}
\par By solving $\bV = \Gamma \odot \psi$, we recover $\psi_{1}=0, \psi_{2}=1$. In summary, the missing values in $\overline{\Psi} = \{(2, 3), (2,4),(5,3),(5,6), (6, 4) \}$ are $0, 1, 0, 0, 1$, respectively.

\subsection{Infeasible recovery}
\label{app:alg:infeasible}
In this section, we highlight a broad class of initial test matrices \( M \) that become unrecoverable once certain entries are missing. A particularly common and problematic scenario arises when the matrix is \emph{dense}, meaning that each test includes a large number of items. Specifically, we derive the following claim
\begin{claim}\label{claim infe}
    Consider a scenario where the measurement matrix has $n$ items, $t$ test, and we know that there are $d$ defectives. Furthermore, assume that among the $t$ tests, there exist test \( t_0 \) in which more than \( n - d \) non-missing entries are 1. In such a case, the missing entries in \( t_0 \) cannot be uniquely determined, leading to our matrix completion problem being infeasible.
\end{claim}
Claim \ref{claim infe} being true is simply due to the pigeonhole principle, for any sample pair \( (\bX, \bY) \), there will always be at least one defective item included in \( t_0 \), resulting in a positive test outcome \( \bY_{t_0} = 1 \), regardless of the actual values of the missing entries. Consequently, these entries are unidentifiable. Within the context of our method, the presence of such a test implies that no matter how we sample \( \bX \), the pair \( (\bX, t_0) \) is not informative. As a result, the missing values in the matrix cannot be fully recovered via group testing, no matter how many columns are there in $\mX$.\\

We will show this in our following example. Consider the matrix $\mM$ defined in~\cref{eqn:connectivity} and its missing matrix as follows:
\begin{equation}
    \mG = \left[ \begin{array}{cccccccccccc}
    0 & \blacksquare & 0 & 1 & 1 & 1 & 1 & 1 & 1 & \blacksquare & 1 & 1 \\
    0 & 0 & 0 & 1 & 1 & 1 & 0 & 0 & 0 & 1 & 0 & 0 \\
    1 & 1 & 1 & 0 & 0 & 0 & 0 & 0 & 0 & 1 & 0 & 0 \\
    0 & 0 & 1 & 0 & 0 & 1 & 0 & 0 & 1 & 0 & 1 & 0 \\
    0 & 1 & 0 & 0 & 1 & 0 & 0 & 1 & 0 & 0 & 1 & 0 \\
    1 & 0 & 0 & 1 & 0 & 0 & 1 & 0 & 0 & 0 & 1 & 0 \\
    0 & 1 & 0 & 0 & 0 & 0 & 1 & 0 & 1 & 0 & 0 & 1 \\
    0 & 1 & 0 & 0 & 1 & 0 & 1 & 0 & 0 & 0 & 0 & 1 \\
    1 & 0 & 0 & 0 & 0 & 1 & 0 & 1 & 0 & 0 & 0 & 1
    \end{array} \right] .
\end{equation}

\par Suppose we know that the number of defectives is $d = 5$. Despite the fact that the matrix \( \mM \) has only two missing entries, it is impossible to recover them due to the overwhelming number of \( 1 \) entries observed in their corresponding tests. Indeed, consider a sample \(\bX \in \operatorname{col}(\mathbf{X})\). Since \(\bX\) must contain at least 5 ones, and the row in \( \mM \) that includes both missing entries has at most entries that are not ones, none of the possible choices of \(\bX\) can yield informative results when paired with \( (\bX, 1) \). Consequently, the converted missing matrix \( \Gamma \) remains empty regardless of the number of samples we collect. In this scenario, recovering \( \mM \) using only the current method is infeasible. Additional conditions or assumptions are required to achieve a solution.
\subsection{Traditional Algorithms for Group Testing} \label{subsec: Tra Al in GT}
In this section, we will briefly go through traditional algorithms for a group testing problem. These four algorithms are COMP, DD, SCOMP, and SSS \cite{aldridge2014group}. Throughout this section, we consider the not-defective set ND as follow:
$$\text{ND} := \left\{ i : \exists t \,  \text{ (a)} \, x_{it} = 1 \, \text{and} \, \text{ (b)} \, y_t = 0  \right\}
$$
and denote the possible defective set PD as the complement of the ND. We further denote $\mathbf{K}$ as the defective set and consider the definition of \textit{satisfying} if 
\begin{definition}
Given a test design $\mathbf{X}$ and outcomes $\mathbf{y}$, we shall call a set of items \( \mathcal{L} \subseteq \mathcal{N} \) a satisfying set if group testing with defective set \( \mathcal{L} \) and test design $\mathbf{X}$ would lead to the outcomes $\mathbf{y}$. Where $\mathcal{N}$ is the set of all item.
\end{definition}
\textbf{Combinatorial Orthogonal Matching Pursuit (COMP). }Introduced by \cite{chan2011non} in 2011, this algorithm treats the items in the set 
ND as non-defective, assuming all other items to be defective. In essence, COMP considers every item as possibly defective. Additionally, the COMP algorithm is prone only to false-positive errors (i.e., classifying non-defective items as defective) and does not make false-negative errors (i.e., it never classifies defective items as non-defective). This implies that we have: $\mathbf{K}\subseteq \mathbf{K}_{\text{COMP}}$. The pseudocode of COMP is shown in \cref{algo: COMP}.
\begin{algorithm}
\caption{The COMP algorithm}
\label{algo: COMP}
\begin{algorithmic}[1]
\Require The measurement matrix $\mX$, the item vector $\mathcal{N}$ and the outcome vector $\mathbf{y}$
\Ensure A vector $\textbf{K}_{COMP}$ representing the defective set.
\State $ND:=\{i:\exists t,\quad x_{it}=1, \quad y_t=0\}$
\State $\textbf{K}_{COMP}\leftarrow PD:=\mathcal{N}\backslash ND $
\end{algorithmic}
\end{algorithm}\\
\textbf{Definite defectives (DD). }After identifying possible defective (PD) items, some elements are confirmed as definitely defective (DD). The main idea is that if a test contains only one PD item, that item must be defective. This motivates the DD algorithm, which starts by using the possible PDs identified in the COMP algorithm. The DD algorithm follows three steps: (1) Define the possible defectives $PD = ND^C$; (2) For each test with a single PD item, mark that item as defective; (3) Declare all remaining items as non-defective. Formally, the DD algorithm defines every item in the set $$DD := \{i \in PD : x_{it} = 1, x_{jt} = 0 \, \forall j \in PD \setminus \{i\}, y_t = 1\}$$ as defective.

Steps 1 and 2 are mistake-free, isolating non-defective (ND) items that are ignored thereafter. However, Step 3 may err by misidentifying masked defectives as non-defective. While the DD algorithm never generates false positives, it can make false negatives, meaning $\mathbf{K}_{\text{DD}} \subseteq \mathbf{K}$. Step 3’s approach is based on the assumption that defectives are rare, allowing us to infer that items in $PD$ but not in $DD$ are non-defective unless proven otherwise. The pseudocode for the DD algorithm is shown in \cref{alg:DD}.

\begin{algorithm}
\caption{The DD algorithm}
\label{alg:DD}
\begin{algorithmic}[1]
\Require The measurement matrix $\mathbf{X}$, the item set $\mathcal{N}$ and the outcome vector $\mathbf{y}$
\Ensure A vector $\textbf{K}_{DD}$ representing the defective set.
\State $ND:=\{i:\exists t,\quad x_{it}=1, \quad y_t=0\}$
\State $PD:=\mathcal{N}\backslash ND$
\State $\textbf{K}_{DD}\leftarrow DD:=\{i\in PD:\exists t, \quad x_{it}=1, \quad x_{jt=0} \forall j\in PD\backslash \{i\}, \quad y_t=1\}$
\end{algorithmic}
\end{algorithm}
\textbf{Sequential COMP. }In the COMP algorithm, the key insight is that $\mathbf{K}_{DD}$ need not be a satisfying set, as some positive tests may contain no elements from $\mathbf{K}_{DD}$. We define a positive test as \textit{unexplained} by $\mathbf{K}$ if it contains no elements from $\mathbf{K}$. A set $\mathbf{K} \subseteq PD$ is a satisfying set if it explains all positive tests. The SCOMP algorithm uses this principle to greedily extend $\mathbf{K}_{DD}$ by adding items from $PD$ that explain the most unexplained tests. This iterative approach updates $\mathbf{K}$ each time an item is added.

The algorithm proceeds as follows: (1) First, apply the DD algorithm’s initial two steps to generate an initial estimate $\mathbf{K}_{DD} = DD$. (2) Given an estimate $\mathbf{K}$, if $\mathbf{K}$ is satisfying, terminate the algorithm and use $\mathbf{K}$ as the final estimate. If not, find the item $i \in PD$ that appears in the largest number of unexplained tests, add it to $\mathbf{K}$, and repeat. This greedy approach iteratively refines $\mathbf{K}$ and outperforms the DD algorithm on which it is based on, and performs very close to optimal. The pseudocode for the SCOMP algorithm is shown in \cref{alg:SCOMP}.

\begin{algorithm}
\caption{The SCOMP algorithm}
\label{alg:SCOMP}
\begin{algorithmic}[1]
\Require The measurement matrix $\mathbf{X}$, the item vector $\mathcal{N}$ and the outcome vector $\mathbf{y}$
\Ensure A vector $\textbf{K}_{SCOMP}$ representing the defective set.
\State $ND:=\{i:\exists t,\quad x_{it}=1, \quad y_t=0\}$
\State $PD:=\mathcal{N}\backslash ND$
\State $\textbf{K} \leftarrow DD:=\{i\in PD:\exists t, \quad x_{it}=1, \quad x_{jt=0} \forall j\in PD\backslash \{i\}, \quad y_t=1\}$
\While{\textbf{K} is not satisfying}
    \State i:=possible defective that appears in the largest number of test unexplained by \textbf{K}
    \State $\textbf{K}\leftarrow \textbf{K}\cup \{i\}$
\EndWhile
\State $\textbf{K}_{SCOMP}:=\textbf{K}$
\end{algorithmic}
\end{algorithm}
\textbf{Smallest Satisfying Set (SSS). }This is an optimal detection algorithm that focuses on the computational feasibility of detecting the true defective set \( \mathbf{K} \). We know that:

\begin{itemize}
    \item \( \mathbf{K} \) is a satisfying set, as we assume noiseless testing.
    \item \( \mathbf{K} \) is likely small since we are dealing with cases where \( |\mathbf{K}| \ll |\mathcal{N}| \).
\end{itemize}

To find the smallest set satisfying the outcomes, we formulate the following 0-1 linear program, akin to compressed sensing:

\begin{align}\label{eq: SSS las}
    \text{minimize} \quad & \mathbf{1}^\top \mathbf{z} \noindent \\
    \text{subject to} \quad & x_t \cdot \mathbf{z} = 0 \quad \text{for} \quad y_t = 0,\noindent \\
                             & x_t \cdot \mathbf{z} \geq 1 \quad \text{for} \quad y_t = 1, \noindent\\
                             & \mathbf{z} \in \{0, 1\}^N.
\end{align}

The smallest satisfying set is found by the SSS algorithm: 
\[
\mathbf{K}_{SSS} = \{ i : z_i = 1 \}.
\]

If the number of defectives \( s = |\mathbf{K}| \) is known, we add the constraint \( \mathbf{1}^\top z \geq s \) to ensure a satisfying set of size exactly \( s \). In this case, the algorithm finds an arbitrary satisfying set of the correct size, making it optimal under the assumption of noiseless testing. For the unknown \( s \) scenario, we consider the SSS algorithm as "essentially optimal". he SSS algorithm is based on the integer linear program used to solve the linear system in \cref{eq: SSS las}.

\bibliographystyle{ieeetr}
\balance
\bibliography{paper}

\end{document}